\documentclass[hidelinks, runningheads]{llncs}
\pdfoutput=1

\newif\ifanonym

\newif\iftikzcompile
\tikzcompiletrue 

\newif\ifsubmit

\usepackage{amsmath,amssymb,amsfonts}       
\usepackage{mathbbol}
\usepackage{bm}
\usepackage{bbm}
\DeclareSymbolFontAlphabet{\amsmathbb}{AMSb}
\usepackage{mathtools}

\usepackage{siunitx}

\usepackage{accents}
\newcommand\munderbar[1]{%
  \underaccent{\bar}{#1}}
  
\newcommand\smallpar[1]{%
\medskip\noindent\emph{#1}}

\newcommand\smallsubsubsection[1]{%
\medskip\noindent\textbf{#1}}

\usepackage[acronym]{glossaries}

\usepackage{fontawesome5}

\usepackage[dvipsnames]{xcolor}
\usepackage{graphicx}
\usepackage{xspace}
\usepackage{microtype}
\usepackage{todonotes}
\usepackage{colonequals}
\usepackage[noadjust,nospace]{cite}
\usetikzlibrary{patterns,arrows,arrows.meta,calc,shapes,shadows,decorations.pathmorphing,decorations.pathreplacing,automata,shapes.multipart,positioning,shapes.geometric,fit,circuits,trees,shapes.gates.logic.US,fit, shadows.blur,shapes.symbols}

\usetikzlibrary{backgrounds,scopes}
\usepackage{marvosym}
\usepackage{multirow}
\usepackage{multicol}

\usepackage{paralist}
\usepackage{tabularx}
\usepackage{booktabs}
\usepackage{scrextend}
\usepackage{threeparttable}

\usepackage{caption}
\usepackage{subcaption}
\usepackage{url}
\usepackage{appendix}

\usepackage[pdfpagelabels=true]{hyperref}
\usepackage{cleveref}

\usepackage{amsthm}

\crefname{equation}{Eq.}{Eq.}
\crefname{pluralequation}{Eqs.}{Eqs.}

\crefname{figure}{Fig.}{Fig.}
\crefname{pluralfigure}{Figs.}{Figs.}

\crefname{section}{Sect.}{Sect.}
\crefname{pluralsection}{Sects.}{Sects.}

\crefname{appendix}{App.}{App.}
\crefname{pluralappendix}{Apps.}{Apps.}

\crefname{table}{Tab.}{Tab.}
\crefname{pluraltable}{Tabs.}{Tabs.}

\crefname{definition}{Def.}{Def.}
\crefname{pluraldefinition}{Defs.}{Defs.}

\crefname{theorem}{Theorem}{Theorems}
\crefname{pluraltheorem}{Theorems}{Theorems}

\crefname{lemma}{Lemma}{Lemma}
\crefname{plurallemma}{Lemmas}{Lemmas}

\crefname{example}{Example}{Example}
\crefname{pluralexample}{Examples}{Examples}

\crefname{assumption}{Assumption}{Assumption}
\crefname{pluralassumption}{Assumptions}{Assumptions}

\crefname{remark}{Remark}{Remark}
\crefname{pluralremark}{Remarks}{Remarks}

\usepackage{tikz} 
\usepackage{pgfplots}
\usepgfplotslibrary{external, fillbetween}
\pgfplotsset{compat=1.8}

\definecolor{red}{rgb}{0.745,0.192,0.102}
\definecolor{darkgreen}{RGB}{34,161,55}
\definecolor{ruhuisstijlrood}{rgb}{0.745,0.192,0.102}
\definecolor{ruhuisstijlzwart}{rgb}{0,0,0}
\definecolor{ruhuisstijlwit}{rgb}{0.98,0.98,0.98}

\usepackage{mdframed}
\usepackage{enumitem}

\makeatletter%
\g@addto@macro\normalsize{%
  \setlength\abovedisplayskip{6.4pt}%
  \setlength\belowdisplayskip{6.3pt}%
}%
\makeatother

\usetikzlibrary{arrows,decorations.pathmorphing,positioning,fit,trees,shapes,shadows,automata,calc}

\newcommand{\ie}{i.e.\@\xspace}

\newcommand{\cf}{cf.\@\xspace}

\newcommand{\storm}{\textrm{Storm}\xspace}

\DeclareMathOperator*{\minimize}{minimize}
\usepackage{array,booktabs}
\newcolumntype{C}{>{$\displaystyle}c<{$}}
\usepackage[normalem]{ulem}
\useunder{\uline}{\ul}{}

\newcommand{\PP}{\amsmathbb{P}}

\newcommand{\probdist}{\ensuremath{\amsmathbb{P}}}
\newcommand{\pr}{\ensuremath{\mathrm{Pr}}}

\newcommand{\colorpar}[3]{\colorbox{#1}{\parbox{#2}{#3}}}
\newcommand{\marginremark}[3]{\marginpar{\colorpar{#2}{3.5em}{\color{#1}#3}}}

\ifsubmit
  \newcommand{\nj}[1]{}
  \newcommand{\sj}[1]{}
  \newcommand{\ms}[1]{}
  \newcommand{\mv}[1]{}
  \newcommand{\tb}[1]{}
  \newcommand{\marielle}[1]{}
\else
  \newcommand{\nj}[1]{\marginremark{yellow}{black}{\tiny{[NJ]~#1}}}
  \newcommand{\sj}[1]{\marginremark{purple}{white}{\tiny{[SJ]~ #1}}}
  \newcommand{\ms}[1]{\marginremark{white}{blue}{\scriptsize{[MS]~ #1}}}
  \newcommand{\mv}[1]{\marginremark{white}{black}{\tiny{[MV]~ #1}}}
  \newcommand{\tb}[1]{\marginremark{black}{yellow}{\tiny{[TB]~ #1}}}
  \newcommand{\marielle}[1]{\todo[inline,caption={},color=yellow!40]{{\it Marielle:~}#1}}
\fi

\makeatletter
\def\THICKhrulefill{\leavevmode \leaders \hrule height 5pt\hfill \kern \z@}
\def\getfirst#1#2\relax{\tctestifnum{\count@stringtoks{#1}>1}{ERROR}{#1}}
\makeatother

\colorlet{MS-fg}{WildStrawberry}
\colorlet{MS-bg}{Plum!6}

\newcommand{\R}{\amsmathbb{R}}
\newcommand{\Q}{\amsmathbb{Q}}
\newcommand{\Z}{\amsmathbb{Z}}
\newcommand{\N}{\amsmathbb{N}}

\newcommand{\distr}[1]{\ensuremath{\mathit{Dist(#1)}}}

\newcommand{\Paramvar}{\ensuremath{{V}}\xspace}        

\newcommand{\sinit}{s_{\mathit{I}}} 
\newcommand{\rateFun}{\ensuremath{\mathbf{R}}}

\newcommand{\paramspace}{\ensuremath{\mathcal{V}_{\mathcal{M}}}}

\newcommand{\instctmc}[1]{\mathcal{M}[#1]}
\newcommand{\pctmc}{\mathcal{M}}

\newcommand{\pCTMC}{\ensuremath{(S, \sinit, \Paramvar, \rateFun})}
\newcommand{\upCTMC}{\ensuremath{(\pctmc, \probdist)}}

\renewcommand{\Pr}{\ensuremath{\textnormal{Pr}}}

\newcommand{\sol}{\ensuremath{\mathsf{sol}}}
\newcommand{\satprob}{\ensuremath{\mathsf{contain}}}

\newcommand{\probc}{\ensuremath{\mathsf{probC}}}
\newcommand{\thresh}{\ensuremath{\mathsf{num}}}

\DeclareMathAlphabet{\mathpzc}{OT1}{pzc}{m}{it}
\def\presuper#1#2%
  {\mathop{}%
   \mathopen{\vphantom{#2}}^{#1}%
   \kern-\scriptspace%
   #2}

\newcommand{\sampleset}{\ensuremath{\mathcal{U}_n}}
\newcommand{\card}[1]{\ensuremath{\vert #1 \vert}}

\usepackage[firstpage]{draftwatermark} 
\SetWatermarkAngle{0}
\SetWatermarkText{\raisebox{12.5cm}{%
  \hspace{0.1cm}%
  \href{https://doi.org/10.5281/zenodo.6523863}{\includegraphics{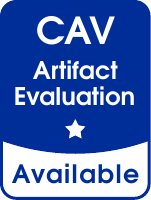}}%
  \hspace{9cm}%
  \includegraphics{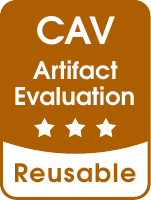}%
}}

\spnewtheorem{assumption}{Assumption}{\bfseries}{\itshape}

\newacronym[plural=DTMCs,firstplural=discrete-time Markov chains]{DTMC}{DTMC}{discrete-time Markov chain}
\newacronym[plural=MDPs,firstplural=Markov decision processes (MDPs)]{MDP}{MDP}{Markov decision process}
\newacronym[plural=iMDPs,firstplural=interval MDPs (iMDPs)]{iMDP}{iMDP}{interval MDP}
\newacronym[plural=uMDPs,firstplural=uncertain MDPs (uMDPs)]{uMDP}{uMDP}{uncertain MDP}
\newacronym[plural=CTMCs,firstplural=continuous-time Markov chains (CTMCs)]{CTMC}{CTMC}{continuous-time Markov chain}
\newacronym[plural=pCTMCs,firstplural=parametric CTMCs (pCTMCs)]{pCTMC}{pCTMC}{parametric CTMC}
\newacronym[plural=upCTMCs,firstplural=uncertain pCTMCs (upCTMCs)]{upCTMC}{upCTMC}{uncertain pCTMC}
\newacronym[plural=POMDPs,firstplural=partially observable Markov decision processes (POMDPs)]{POMDP}{POMDP}{partially observable Markov decision process}
\newacronym[plural=DFTs,firstplural=dynamic fault trees (DFTs)]{DFT}{DFT}{dynamic fault tree}

\newacronym[]{iid}{i.i.d.}{independent and identically distributed}

\newacronym[]{CSL}{CSL}{continuous stochastic logic}

\definecolor{plotblue}{rgb}{0.1,0.498039215686275,0.9549019607843137}

\newif\ifappendix
\global\appendixtrue
\makeatletter
\def\orcidID#1{\smash{\href{http://orcid.org/#1}{\protect\raisebox{-1.25pt}{\protect\includegraphics{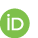}}}}}
\makeatother

\begin{document}
\title{Sampling-Based Verification of CTMCs \\ with Uncertain~Rates
\ifanonym\else
\thanks{This work has been partially funded by NWO under the grant \href{https://primavera-project.com}{PrimaVera}, number NWA.1160.18.238, and by EU Horizon 2020 project MISSION, number 101008233.}
\fi
}
\titlerunning{Sampling-Based Verification of CTMCs with Uncertain~Rates}

\ifanonym
\author{Authors omitted.}
\institute{}
\authorrunning{Authors omitted.}
\else
\author{
Thom S. Badings\inst{1}\orcidID{0000-0002-5235-1967}
\and Nils Jansen\inst{1}\orcidID{0000-0003-1318-8973}
\and Sebastian Junges\inst{1}\orcidID{0000-0003-0978-8466}
\and\\ Marielle Stoelinga\inst{1,2}\orcidID{0000-0001-6793-8165}
\and Matthias Volk\inst{2}\orcidID{0000-0002-3810-4185}
}
\authorrunning{T.S. Badings, N. Jansen, S. Junges, M.I.A. Stoelinga, and M. Volk}
\institute{
Radboud University, Nijmegen, the Netherlands \\ \email{thom.badings@ru.nl} 
\and
University of Twente, Enschede, the Netherlands
}
\fi

\maketitle
\begin{abstract}
We employ uncertain parametric CTMCs with parametric transition rates and a prior on the parameter values. The prior encodes uncertainty about the actual transition rates, while the parameters allow dependencies between transition rates. Sampling the parameter values from the prior distribution then yields a standard CTMC, for which we may compute relevant reachability probabilities. We provide a principled solution, based on a technique called scenario-optimization, to the following problem: From a finite set of parameter samples and a user-specified confidence level, compute prediction regions on the reachability probabilities. The prediction regions should (with high probability) contain the reachability probabilities of a CTMC induced by any additional sample. To boost the scalability of the approach, we employ standard abstraction techniques and adapt our methodology to support approximate reachability probabilities. Experiments with various well-known benchmarks show the applicability of the approach.
\end{abstract}

\section{Introduction}
\label{sec:Introduction}

\noindent
\Glspl{CTMC} are widely used to model complex probabilistic systems in reliability engineering~\cite{DBLP:journals/csr/RuijtersS15}, network processes~\cite{DBLP:conf/srds/HaverkortHK00,hermanns1999multi}, systems biology~\cite{DBLP:journals/acta/CeskaDPKB17,DBLP:conf/tacas/BortolussiS18} and epidemic modeling~\cite{ALLEN2017128}.
A key verification task is to compute aspects of system behavior from these models, expressed as, e.g., continuous stochastic logic (CSL) formulae~\cite{DBLP:journals/tse/BaierHHK03,aziz2000model}.
Typically, we compute reachability probabilities for a set of horizons, such as: \emph{what is the probability that a target state is reached before time $t_1,\ldots,t_n$?}
Standard algorithms~\cite{DBLP:journals/tse/BaierHHK03} implemented in mature model checking tools such as \storm~\cite{StormSTTT} or Prism~\cite{DBLP:conf/cav/KwiatkowskaNP11} provide efficient means to compute these \emph{reachability probabilities}.
However, these methods typically require that transition rates and probabilities are precisely known. 
This assumption is often unrealistic~\cite{DBLP:conf/rtss/HanKM08} and led to some related work, which we discuss in \cref{sec:Related}.

\smallpar{Illustrative example.}
An epidemic can abstractly be modeled as a finite-state \gls{CTMC}, e.g., the SIR (susceptible-infected-recovered) model~\cite{andersson2012stochastic}, which is shown in \cref{fig:sirctmc} for a population of two.
Such a \gls{CTMC} assumes a \emph{fixed set of transition rates}, in this case an infection rate $\lambda_i$, and a recovery rate $\lambda_r$. 
The outcome of analyzing this \gls{CTMC} for fixed values of $\lambda_i$ and $\lambda_r$ may yield a \emph{probability curve} like in Fig.~\ref{fig:intro1a_singlecurve}\footnote{For visual clarity, we plot the reachability probability \emph{between} time 100 and $t_1,\hdots,t_n$.}, where we plot the probability (y-axis) of reaching a target state that corresponds to the epidemic becoming extinct against varying time horizons (x-axis).
In fact, the plot is obtained via a smooth interpolation of the results at finitely many horizons, cf.~\ref{fig:intro1b_curveaspoints}.
To acknowledge that $\lambda_i, \lambda_r$ are in fact unknown, we may analyze the model for different values of $\lambda_i, \lambda_r$, resulting in a set of curves as in Fig.~\ref{fig:intro1b_multiplecurves}.
These individual curves, however, provide no guarantees about the shape of the curve obtained from another infection and recovery rate. Instead, we assume a \emph{probability distribution} over the transition rates and aim to compute \emph{prediction regions} as those in shown Fig.~\ref{fig:intro1c_confidence}, in such a way that with a certain (high) probability, any rates $\lambda_i$ and $\lambda_r$ yield a curve within this region.

\smallpar{Overall goal.} 
From the illustrative example, we state the following goal. 
Each fixed set of transition rates induces a \emph{probability curve}, i.e., a mapping from horizons to the corresponding reachability probabilities.
We aim to construct \emph{prediction regions} around a set of probability curves, such that with high probability and high confidence, sampling a set of transition rates induces a probability curve within this region.
Our key contribution is an efficient \emph{probably approximately correct}, or PAC-style method that computes these prediction regions. 
The remainder of the introduction explores the technical steps toward this goal. 

\begin{figure}[t]
\begin{subfigure}[b]{0.47\linewidth}
\iftikzcompile
  \begin{tikzpicture}[node distance=0.6cm]
      \node[state, inner sep=1pt] (SR) {\footnotesize{SR}};
      \node[state,right=of SR, initial,initial where=above, initial text=] (SI) {SI};
      \node[state,right=of SI] (II) {\footnotesize{II}};
      \node[state,right=of II] (RI) {\footnotesize{RI}};
      \node[state,right=of RI] (RR) {\footnotesize{RR}};
      \draw[->] (SI) -- node[above] {$\lambda_r$} (SR);
      \draw[->] (SI) -- node[above] {$\lambda_i$} (II);
      \draw[->] (II) -- node[above] {$2\lambda_r$} (RI);
    \draw[->] (RI) -- node[above] {$\lambda_r$} (RR);
      
  \end{tikzpicture}
\fi
    \caption{pCTMC $\pctmc$ with parameters ${\lambda_i, \lambda_r}$.}
    \label{fig:sirctmc}
\end{subfigure}
\hfill 
\begin{subfigure}[b]{0.47\linewidth}
\centering
\iftikzcompile
    \def\centerx{0}
    \def\centery{0}
    \begin{tikzpicture}
    \begin{axis}[
        width=.73\linewidth,
        height=1in,
        xlabel={$\lambda_i$},
        ylabel={$\lambda_r$},
        xtick=\empty,
        ytick=\empty,
        zmin=0,
        zmax=1,
        ztick={0,1},
        x label style={yshift=0.1cm},
        y label style={yshift=0.1cm},]
    \addplot3[surf,domain=-4:4,domain y=-4:4] 
        {exp(-( (x-\centerx)^2 + (y-\centery)^2)/3 )};
    \end{axis}
    \end{tikzpicture}
\fi
    \caption{Distribution $\PP$ over  values for $({\lambda_i, \lambda_r})$.}
    \label{fig:sirdistribution}
\end{subfigure}
    \caption{An upCTMC $(\pctmc, \PP)$ for the SIR (pop=2) model.}
\end{figure}

\begin{figure}[t!]
\centering
\begin{subfigure}[b]{0.24\linewidth}
\iftikzcompile
    \scalebox{1.0}{

\begin{tikzpicture}[baseline]
\definecolor{color1}{rgb}{0.1,0.498039215686275,0.9549019607843137}
\begin{axis}[
        width=1\linewidth,
        height=1.45in,
        ymin=-0.001,
        ymax=1.001,
        xmin=100,
        xmax=200,
        xlabel={Time (weeks)},
        ylabel={Probability},
        xtick={100,150,200},
        x label style={at={(axis description cs:0.5,-0.13)}},
]

\addplot[thick, mark=none, color=plotblue] table[x=t, y=x44, col sep=comma] {Figures/Tikz/Data/epidemic_curves100.csv};

\end{axis}
\end{tikzpicture}}%
\fi
    \captionof{figure}{Curve for a single \gls{CTMC} with precise transition rates.}
    \label{fig:intro1a_singlecurve}
\end{subfigure}
\hfill
\begin{subfigure}[b]{0.22\linewidth}
\iftikzcompile
    \scalebox{1.0}{

\begin{tikzpicture}[baseline]
\definecolor{color1}{rgb}{0.1,0.498039215686275,0.9549019607843137}
\begin{axis}[
        width=1.3\linewidth,
        height=1.45in,
        ymin=-0.001,
        ymax=1.001,
        xmin=100,
        xmax=200,
        xlabel={Time (weeks)},
        ylabel=\empty,
        xtick={100,150,200},
        ytick={0,0.2,0.4,0.6,0.8,1.0},
        yticklabels=\empty,
        x label style={at={(axis description cs:0.5,-0.13)}},
]

\addplot[thick, mark=o, color=plotblue, draw=none] table[x=t, y=x44, col sep=comma] {Figures/Tikz/Data/epidemic_curves100_coarse.csv};

\end{axis}
\end{tikzpicture}}%
\fi
    \captionof{figure}{Point abstraction of a curve for a single \gls{CTMC}.}
    \label{fig:intro1b_curveaspoints}
\end{subfigure}
\hfill
\begin{subfigure}[b]{0.22\linewidth}
\iftikzcompile
    \scalebox{1.0}{

\begin{tikzpicture}[baseline]
\definecolor{color1}{rgb}{0.1,0.498039215686275,0.9549019607843137}
\begin{axis}[
        width=1.3\linewidth,
        height=1.45in,
        ymin=-0.001,
        ymax=1.001,
        xmin=100,
        xmax=200,
        xlabel={Time (weeks)},
        xtick={100, 130, 150, 170, 200},
        xticklabels={, $t_1$, , $t_2$, },
        extra x ticks={130, 170},
        extra x tick labels={},
        extra tick style={
            grid=major,
            grid style=dashed
        },
        ylabel=\empty,
        ytick={0,0.2,0.4,0.6,0.8,1.0},
        yticklabels=\empty,
        x label style={at={(axis description cs:0.5,-0.13)}},
]

\addplot[thick, mark=none, color=plotblue] table[x=t, y=x44, col sep=comma] {Figures/Tikz/Data/epidemic_curves100.csv};
\addplot[thick, mark=none, color=BurntOrange] table[x=t, y=x7, col sep=comma] {Figures/Tikz/Data/epidemic_curves100.csv};
\addplot[thick, mark=none, color=OliveGreen] table[x=t, y=x11, col sep=comma] {Figures/Tikz/Data/epidemic_curves100.csv};
\addplot[thick, mark=none, color=Mahogany] table[x=t, y=x13, col sep=comma] {Figures/Tikz/Data/epidemic_curves100.csv};
\addplot[thick, mark=none, color=Fuchsia] table[x=t, y=x14, col sep=comma] {Figures/Tikz/Data/epidemic_curves100.csv};

\end{axis}
\end{tikzpicture}}%
\fi
    \captionof{figure}{Curves for five \glspl{CTMC} with different rates.}
    \label{fig:intro1b_multiplecurves}
\end{subfigure}
\hfill
\begin{subfigure}[b]{0.24\linewidth}
\iftikzcompile
    \scalebox{1.0}{

\begin{tikzpicture}[baseline]
\definecolor{color1}{rgb}{0.1,0.498039215686275,0.9549019607843137}
\begin{axis}[
        width=1.3\linewidth,
        height=1.45in,
        ymin=0,
        ymax=1,
        xmin=100,
        xmax=200,
        grid style=dashed,
        xlabel={Time (weeks)},
        xtick={100, 130, 150, 170, 200},
        xticklabels={, $t_1$, , $t_2$, },
        extra x ticks={130, 170},
        extra x tick labels={},
        extra tick style={
            grid=major,
            grid style=dashed
        },
        ylabel=\empty,
        ytick={0,0.2,0.4,0.6,0.8,1.0},
        yticklabels=\empty,
        legend entries={High prob.,Low prob.},
        legend pos=north west,
        legend image post style={xscale=0.3},
        legend style={nodes={scale=0.75, transform shape},
                      row sep=0pt},
        x label style={at={(axis description cs:0.5,-0.13)}},
]

\addplot[name path=A, const plot, mark=none, color=red!80, forget plot] table[x=t, y=xlow_large, col sep=comma] {Figures/Tikz/Data/epidemic_confidenceRegions.csv};

\addplot[name path=B, const plot, mark=none, color=red!80, forget plot] table[x=t, y=xupp_large, col sep=comma] {Figures/Tikz/Data/epidemic_confidenceRegions.csv};

\addplot[fill=red!15,opacity=1] fill between[of=A and B,split,soft clip={domain=100:200},every segment no 0/.style={red!15},];


\addplot[name path=C, const plot, mark=none, color=MidnightBlue!80, forget plot] table[x=t, y=xlow_small, col sep=comma] {Figures/Tikz/Data/epidemic_confidenceRegions.csv};

\addplot[name path=D, const plot, mark=none, color=MidnightBlue!80, forget plot] table[x=t, y=xupp_small, col sep=comma] {Figures/Tikz/Data/epidemic_confidenceRegions.csv};

\addplot[fill=MidnightBlue!30,opacity=1] fill between[of=C and D,split,soft clip={domain=100:200},every segment no 0/.style={MidnightBlue!30},];

\end{axis}

\end{tikzpicture}}%
\fi
    \captionof{figure}{Two prediction regions with different probabilities.}
    \label{fig:intro1c_confidence}
\end{subfigure}

\caption{
The probability of extinction in the SIR (140) model for horizons $[100,t]$.
}
\label{fig:intro1}

\end{figure}
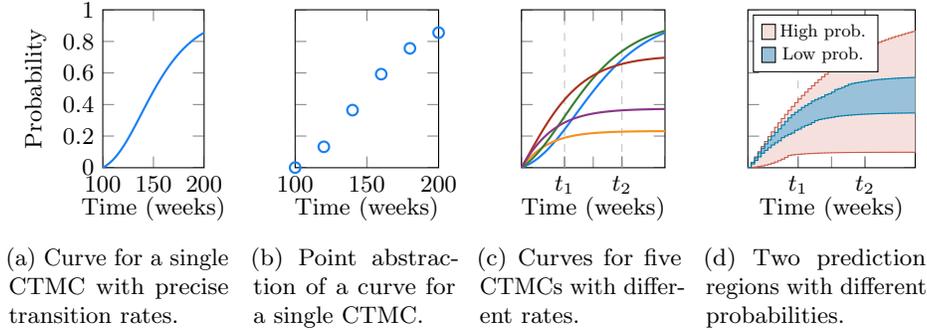

\smallpar{Uncertain CTMCs.}
The setting above is formally captured by \glspl{pCTMC}.
Transition rates of \glspl{pCTMC} are not given precisely but as (polynomials over) parameters~\cite{DBLP:conf/rtss/HanKM08, DBLP:journals/jss/CalinescuCGKP18}, such as those shown in \cref{fig:sirctmc}.
We assume a \emph{prior} on each parameter valuation, i.e., assignment of values to parameters, similar to settings in~\cite{DBLP:conf/tacas/BortolussiS18,DBLP:journals/sosym/MeedeniyaMAG14} and in contrast to, e.g., \cite{DBLP:conf/rtss/HanKM08,DBLP:journals/acta/CeskaDPKB17}.
These priors may result from asking different experts which value they would assume for, e.g., the infection rate.
The prior may also be the result of Bayesian reasoning~\cite{DBLP:conf/qest/WijesuriyaA19}.
Formally, we capture the uncertainty in the rates by an arbitrary and potentially unknown \emph{probability distribution} over the parameter space, see \cref{fig:sirdistribution}.
We call this model an \emph{\gls{upCTMC}}. 
The distribution allows drawing \gls{iid} \emph{samples} that yield (parameter-free) CTMCs. 

\smallpar{Problem statement.}
We consider prediction regions on probability curves in the form of a pair of two curves that `sandwich' the probability curves, as depicted in~\cref{fig:intro1c_confidence}. 
Intuitively, we then aim to find a prediction region $R$ that is sufficiently large, such that  sampling parameter valuations yields a probability curve in $R$ with high probability $p$.
We aim to compute a lower bound on this  \emph{containment probability} $p$. 
Naturally, we also aim to compute a meaningful, i.e. small (tight), prediction region $R$.
As such, we aim to solve the following problem:
\begin{mdframed}[backgroundcolor=gray!30, nobreak=true]
    \textbf{Problem Statement.} \, Given a upCTMC with a target state, compute
    \begin{enumerate}[noitemsep,topsep=1pt]
        \item a (tight) \emph{prediction region} $R$ on the probability curves, and 
        \item a (tight) \emph{lower bound on the containment probability} that a sampled parameter valuation induces a probability curve that will lie in $R$.
    \end{enumerate}
We solve this problem with a user-specified confidence level $\beta$.
\end{mdframed}%
\smallpar{The problem solved.}
In this paper, we present a method that samples probability curves as in \cref{fig:intro1b_multiplecurves}, but now for, say 100 curves. 
From these curves, we compute prediction regions (e.g., both tubes in \cref{fig:intro1c_confidence}) and compute a lower bound (one for both tubes) on the containment probability that the curve associated with any sampled parameter value will lie in the specific prediction region (tube). 
Specifically, for a confidence level of 99\% and considering 100 curves, we conclude that this lower bound is $79.4\%$ for the red region and $7.5\%$ for the blue region.
For a higher confidence level of $99.9\%$, the lower bounds are slightly more conservative.

\begin{figure}[t]
\begin{subfigure}[b]{0.49\linewidth}
\iftikzcompile
    \pgfdeclarelayer{bg}    
\pgfsetlayers{bg,main}  

\begin{tikzpicture}
\begin{axis}[%
    width=\linewidth,
    height=3.5cm,
    ymin=0,
    ymax=1,
    xmin=0,
    xmax=1.05,
    xlabel={Probability($t_1$)},
    ylabel={Probability($t_2$)},
    xtick={0,0.2,0.4,0.6,0.8,1.0},
    legend entries={{High prob.}, {Low prob.}},
    legend pos=north east,
    legend image post style={xscale=0.55},
    legend style={nodes={scale=0.8, transform shape}},
    x label style={at={(axis description cs:0.5,-0.15)}},
]

\addlegendimage{dashed, very thick, color=red}
\addlegendimage{dotted, very thick, color=MidnightBlue}

\addplot[only marks,fill=black,mark size=0.8pt]%
table[] {
x           y
0.244603	0.303837
0.373032	0.508183
0.253979	0.317177
0.119882	0.532342
0.434642	0.688798
0.359122	0.482489
0.188847	0.227787
0.236949	0.679537
0.421995	0.621958
0.133119	0.156049
0.323366	0.737941
0.148773	0.175697
0.43105	    0.657206
0.285282	0.363072
0.309016	0.399272
0.349309	0.465565
0.302393	0.389076
0.435563	0.692422
0.379207	0.519899
0.376495	0.515416
0.285042	0.362431
0.414811	0.742751
0.426036	0.634863
0.330124	0.433433
0.257467	0.322173
0.426915	0.637363
0.406288	0.7476
0.199711	0.242217
0.195808	0.236997
0.365447	0.494573
0.249974	0.311479
0.162269	0.192977
0.381803	0.525337
0.209823	0.25584
0.427284	0.729401
0.310709	0.401329
0.14788	    0.17461
0.250462	0.311832
0.338928	0.447615
0.205548	0.648365
0.162953	0.193899
0.434881	0.680216
0.426177	0.63264
0.241808	0.683438
0.385196	0.532177
0.203913	0.247891
0.123079	0.14353
0.284565	0.360987
0.326756	0.427393
0.0883575	0.100993
0.26774	    0.336986
0.339878	0.449703
0.189459	0.228606
0.423558	0.625905
0.43398	    0.707056
0.358938	0.481876
0.281992	0.357989
0.269236	0.339438
0.265832	0.702905
0.172939	0.206931
0.137585	0.161668
0.181359	0.62189
0.432807	0.663968
0.323913	0.423006
0.26601	    0.333772
0.404842	0.752342
0.235419	0.290999
0.417143	0.606543
0.399359	0.752763
0.196678	0.237864
0.392106	0.547323
0.393136	0.549492
0.131988	0.154529
0.126649	0.147988
0.387039	0.536564
0.223862	0.275046
0.195521	0.637652
0.419881	0.614393
0.272404	0.343919
0.114967	0.133434
0.232978	0.675925
0.220603	0.270619
0.21643	    0.659733
0.211799	0.258609
0.419204	0.610541
0.356016	0.477279
0.0992507	0.114236
0.428003	0.641092
0.253683	0.316623
0.398014	0.751075
0.146558	0.172997
0.368477	0.499663
0.333983	0.439596
0.28435	    0.361434
0.363421	0.754207
0.419947	0.613342
0.385343	0.531729
0.414475	0.74453
0.392075	0.546364
0.0852553	0.0973077
};

\coordinate (low) at (axis cs:0.0852553, 0.0973077) {};
\coordinate (upp) at (axis cs:0.435563, 0.754207) {};
\coordinate (leftupp) at (axis cs:0.0852553, 0.754207) {};
\coordinate (rightlow) at (axis cs:0.435563, 0.0973077) {};

\coordinate (low2) at (axis cs:0.244603, 0.339438) {};
\coordinate (upp2) at (axis cs:0.365447, 0.547323) {};

\end{axis}

\begin{pgfonlayer}{bg}
    \filldraw[dotted, color=red!80, fill=red!10, very thick] (low) rectangle (upp);
\end{pgfonlayer}

\begin{pgfonlayer}{bg}
    \filldraw[dotted, color=MidnightBlue!80, fill=MidnightBlue!20, very thick] (low2) rectangle (upp2);
\end{pgfonlayer}

\end{tikzpicture}
\fi
    \captionof{figure}{Reachability at time points $t_1$ and $t_2$.}
    \label{fig:confidence_2D}
\end{subfigure}
\begin{subfigure}[b]{0.49\linewidth}
\iftikzcompile
    \pgfdeclarelayer{bg}    
\pgfsetlayers{bg,main}  

\begin{tikzpicture}
\begin{axis}[%
    width=\linewidth,
    height=3.5cm,
    ymin=0,
    ymax=1,
    xmin=246.698,
    xmax=430,
    xlabel={Measure 1},
    ylabel={Measure 2},
    x label style={at={(axis description cs:0.5,-0.14)}},
]

\addplot[only marks,fill=black,mark size=1.25pt]%
table[] {
x       y
336.864	0.210376
384.295	0.423492
362.198	0.0274101
384.923	0.125361
386.44	0.0902057
334.481	0.798121
336.477	0.803113
359.125	0.112529
364.513	0.0978145
383.633	0.034566
371.708	0.440136
340.137	0.15877
320.855	0.0615315
354.706	0.614792
353.883	0.00632818
374.746	0.0231702
330.633	0.302067
354.863	0.576432
345.892	0.831048
284.801	0.994567
297.771	0.150905
285.789	0.927085
285.364	0.880599
345.824	0.363983
378.69	0.486016
384.132	0.374522
341.826	0.25646
330.998	0.00876365
351.221	0.81891
366.104	0.0990712
351.436	0.193063
396.855	0.37355
388.159	0.211795
387.549	0.0243667
368.452	0.201633
393.302	0.112361
366.605	0.104177
262.568	0.993929
317.901	0.0603621
399.851	0.0313703
343.807	0.411765
319.64	0.579702
309.891	0.41759
389.786	0.094479
358.893	0.427729
363.704	0.254118
315.114	0.395185
344.666	0.889212
352.24	0.106534
345.423	0.0158268
374.888	0.0553574
368.139	0.00951059
289.567	0.949666
362.013	0.0379405
353.761	0.508427
287.66	0.987665
250.533	0.99844
398.942	0.0178681
387.066	0.00683832
328.309	0.171476
299.356	0.963775
267.105	0.977688
346.299	0.0353068
407.314	0.0189748
391.244	0.450166
304.448	0.7508
329.342	0.100025
398.72	0.132396
337.461	0.355229
408.506	0.0382629
386.442	0.220922
393.049	0.128073
395.29	0.0658804
306.977	0.751666
369.333	0.354568
374.334	0.0127881
395.49	0.0734316
364.715	0.0504726
306.029	0.974674
327.119	0.298003
320.862	0.498231
266.369	0.9977
319.686	0.124312
361.384	0.00429541
395.791	0.0898354
363.023	0.371262
354.595	0.668765
284.441	0.982836
348.108	0.121495
359.951	0.590622
306.554	0.738594
307.865	0.894121
246.798	0.994035
371.221	0.517148
261.938	0.9971
349.334	0.00865675
391.276	0.396833
301.41	0.198117
335.866	0.427223
331.812	0.503027
};

\coordinate (A) at (axis cs:246.698,    0.99844) {};
\coordinate (A2) at (axis cs:246.698,    0) {};
\coordinate (B) at (axis cs:408.506,	0) {};
\coordinate (C) at (axis cs:408.506,	0.0382629) {};
\coordinate (D) at (axis cs:406.523,	0.135794) {};
\coordinate (E) at (axis cs:396.234,	0.388813) {};
\coordinate (F) at (axis cs:366.407,	0.755553) {};
\coordinate (G) at (axis cs:344.666,	0.889212) {};
\coordinate (H) at (axis cs:327.698,	0.941371) {};
\coordinate (I) at (axis cs:301.371,	0.981834) {};
\coordinate (J) at (axis cs:284.801,	0.994567) {};
\coordinate (K) at (axis cs:274.722,	0.99844) {};

\end{axis}

\begin{pgfonlayer}{bg}
    \filldraw[dotted, color=red!80, fill=red!10, very thick] (A) -- (A2) -- (B) -- (C) -- (D) -- (E) -- (F) -- (G) -- (H) -- (I) -- (J) -- (K);
\end{pgfonlayer}

\end{tikzpicture}
\fi
    \captionof{figure}{Pareto front for two measures.}
    \label{fig:confidence_pareto}
\end{subfigure}
\caption{Prediction regions on the solutions vectors for two different \glspl{upCTMC}.}
\label{fig:confidence_both}
\end{figure}
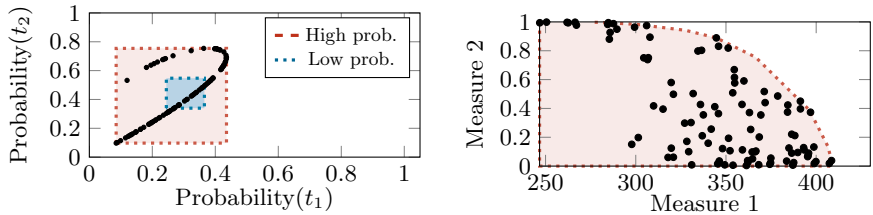

\smallpar{A change in perspective.}
Toward the algorithm, we make a change in perspective.
For two horizons $t_1$ and $t_2$, reachability probabilities for fixed CTMCs are two-dimensional points in $[0,1]^2$ that we call \emph{solution vectors}, as shown in \cref{fig:confidence_2D}.
Here, these solution vectors represent pairs of the probabilities that the disease becomes extinct before time $t_1$ and before $t_2$.
The prediction regions as in \cref{fig:intro1c_confidence} are shown as the shaded boxes in \cref{fig:confidence_2D}.

\smallpar{Solving the problem algorithmically.}
We solve the problem using a sampling-based approach. 
Starting with a set of solution vectors, we use techniques from \emph{scenario optimization}, a data-driven methodology for solving stochastic optimization problems~\cite{campi2018introduction,Campi2021scenarioMaking}. 
As such, we construct the prediction region from the solution to an optimization problem. Our method can balance the size of the prediction region with the containment probability,
as illustrated by the two boxes in \cref{fig:confidence_2D}.

\smallpar{Extensions.} 
Our approach offers more than prediction regions on probability curves from precise samples. 
The change in perspective mentioned above allows for solution vectors that represent \emph{multiple objectives}, such as the reachability with respect to different goal states, expected rewards or even the probability mass of paths satisfying more complex temporal properties. 
In our experiments, we show that this multi-objective approach ---also on probability curves--- yields much tighter bounds on the containment probability than an approach that analyzes each objective independently.
We can also produce prediction regions as other shapes than boxes, as, for example, shown in \cref{fig:confidence_pareto}.
To accelerate our approach, we significantly extend the methodology for dealing with \emph{imprecise verification results}, given as an interval on each entry of the solution vector.

\smallpar{Contributions.} 
Our key contribution is the approach that provides prediction regions and lower bounds on probability curves for upCTMCs. 
The approach requires only about 100 samples and scales to upCTMCs with tens of parameters.
Furthermore:
(1) We extend our approach such that we can also handle the case where only imprecise intervals on the verification results are available.
(2) We develop a tailored batch verification method in the model checker \storm~\cite{StormSTTT} to accelerate the required batches of verification tasks.
We accompany our contributions by a thorough empirical evaluation and remark that our batch verification method can be used beyond scenario optimization.
Our scenario optimization results are independent of the model checking and are, thus, applicable to any model where solution vectors are obtained in the same way as for \glspl{upCTMC}.

\smallpar{Data availability.}
All source code, benchmarks, and logfiles used to produce the data are archived: \url{https://doi.org/10.5281/zenodo.6523863}.
\section{Problem Statement}
\label{sec:Problem}
In this section, we introduce \glspl{pCTMC} and \glspl{upCTMC}, and we define the formal problem statement. 
We use probability distributions over finite and infinite sets; see~\cite{bertsekas2000introduction} for details.
The set of all distributions over a set $X$ is denoted by $\distr{X}$.
The set of polynomials over parameters $\Paramvar$, with rational coefficients, is denoted by $\Q[\Paramvar].$
An \emph{instantiation} $u\colon \Paramvar \rightarrow \Q$ maps parameters to concrete values.
We often fix a parameter ordering and denote instantiations as vectors, $u \in \Q^{|\Paramvar|}$.

\begin{definition}[pCTMC]
    \label{def:pCTMC}
    A \emph{\gls{pCTMC}} is a tuple $\pctmc = \pCTMC$, where $S$ is a finite set of states, $\sinit \in \distr{S}$ is the initial distribution, $\Paramvar$ are the (ordered) parameters, and $\rateFun \colon S \times S \to \Q[\Paramvar]$ is a parametric transition rate function.
    If $\rateFun(s,s) \in \Q_{\geq 0}$ for all $s, s' \in S$, then $\pctmc$ is a (parameter-free) \emph{CTMC}.
\end{definition}

\noindent
For any pair of states $s, s' \in S$ with a non-zero rate $\rateFun(s, s') > 0$, the probability of triggering a transition from $s$ to $s'$ within $t$ time units is $1-e^{-\rateFun(s, s') \cdot  t}$~\cite{DBLP:conf/lics/Katoen16}.

Applying an \emph{instantiation} $u$ to a \gls{pCTMC} $\pctmc$ yields an \emph{instantiated} CTMC $\instctmc{u} = (S, \sinit, \Paramvar, \rateFun[u])$ where $\rateFun[u](s, s') = \rateFun(s, s')[u]$ for all $s, s' \in S$. 
In the remainder, we only consider instantiations $u$ for a \gls{pCTMC} $\pctmc$ which are \emph{well-defined}.
The set of such instantiations is the parameter space $\paramspace$.

A central \emph{measure} on CTMCs is the \emph{(time-bounded) reachability} $\pr(\lozenge^{\leq \tau} E)$, which describes the probability that one of the error states $E$\footnote{Formally, states are labeled and $E$ describes the label, see \cite{BK08}.} is reached within the horizon $\tau \in \Q$.
Other measures include the expected time to reach a particular state, or the average time spent in particular states.
We refer to~\cite{DBLP:conf/lics/Katoen16} for details.

Given a concrete (instantiated) \gls{CTMC} $\instctmc{u}$, the \emph{solution} for measure $\varphi$ is denoted by $\sol_{\instctmc{u}}^\varphi \in \R$; the \emph{solution vector} $\sol_{\instctmc{u}}^\Phi \in \R^m$ generalizes this concept to an (ordered) set of $m$ measures $\Phi = \varphi_1, \dots, \varphi_m$.
We abuse notation and introduce the \emph{solution function} to express solution vectors on a \gls{pCTMC}:
\begin{definition}[Solution function]
    \label{def:solutionFunction}
    A \emph{solution function} $\sol_{\pctmc}^\Phi \colon \paramspace \to \R^{|\Phi|}$ is a mapping from a parameter instantiation $u \in \paramspace$ to the solution vector $\sol_{\instctmc{u}}^\Phi$.
\end{definition}
\noindent
We often omit the scripts in $\sol_{\pctmc}^\Phi(u)$ and write $\sol(u)$ instead.
We also refer to $\sol(u)$ as the solution vector of $u$.
For $n$ parameter samples $\sampleset = \{ u_1, \ldots, u_n \}$ with $u_i \in \paramspace$, we denote the solution vectors by $\sol(\sampleset) \in \R^{m \times n}$.

Using solution vectors, we can define the probability curves shown in \cref{fig:intro1b_multiplecurves}.
\begin{definition}[Probability curve]
    \label{def:probabilityCurve}
    The \emph{probability curve} for reachability probability $\phi_\tau=\pr(\lozenge^{\leq \tau} E)$ and \gls{CTMC} $\instctmc{u}$ is given by $\probc: \tau \mapsto \sol_{\instctmc{u}}^{\varphi_\tau}$.
\end{definition}
\noindent 
We can approximate the function $\probc$ for a concrete \gls{CTMC} by computing $\probc(t_1), \dots, \probc(t_m)$ for a finite set of time horizons.
As such, we compute the solution vector w.r.t. $m$ different reachability measures $\Phi = \{\varphi_{t_1}, \ldots, \varphi_{t_m}\}$.
By exploiting the monotonicity\footnote{In \cref{def:probabilityCurve}, only the upper limit on the timebound is varied, so measures are monotonic.} of the reachability over time, we obtain an upper and lower bound on $\probc(\tau)$ as two step functions, see \cref{fig:intro1c_confidence}.
We can smoothen the approximation, by taking an upper and lower bound on these step functions.

We study \glspl{pCTMC} where the parameters follow a probability distribution. 
This probability distribution can be highly complex or even unknown; we merely assume that we can sample from this distribution.
\begin{definition}[upCTMC]
    \label{def:upCTMC}
    A \emph{\gls{upCTMC}} is a tuple $\upCTMC$ with $\pctmc$ a \gls{pCTMC} and $\probdist$ a probability distribution over the parameter space $\paramspace$ of $\pctmc$.
\end{definition}
\noindent
A \gls{upCTMC} defines a probability space $(\paramspace, \PP)$ over the parameter values, whose domain is defined by the parameter space $\paramspace$.
In the remainder, we denote a \emph{sample} from $\paramspace$ drawn according to $\PP$ by $u \in \paramspace$.

To quantify the performance of a \gls{upCTMC}, we may construct a \emph{prediction region} on the solution vector space, such as those shown in \cref{fig:confidence_2D}.
In this paper, we consider only prediction regions which are compact subsets $R \subseteq \R^{|\Phi|}$.
We define the so-called \emph{containment probability} of a prediction region, which is the probability that the solution vector $\sol(u)$ for a randomly sampled parameter $u \in \paramspace$ is contained in $R$, as follows:
\begin{definition}[Containment probability]
    \label{def:ContainmentProbability}
    For a prediction region $R$, the \emph{containment probability} $\satprob_\mathcal{V}(R)$ is the probability that the solution vector $\sol(u)$ for any parameter sample $u \in \paramspace$ is contained in $R$:
    \begin{equation}
        \label{eq:ContainmentProbability}
        \satprob_\mathcal{V}(R) = \Pr \{ 
            u \in \paramspace \,\colon\, 
            \sol(u) \in R
        \}.
    \end{equation}
\end{definition}
\noindent
Recall that we solve the problem in \cref{sec:Introduction} with a user-specified confidence level, denoted by $\beta \in (0,1)$.
Formally, we solve the following problem:
\begin{mdframed}[backgroundcolor=gray!30]
\textbf{Formal Problem.} \, Given a \gls{upCTMC} $\upCTMC$, a set $\Phi$ of measures, and a confidence level $\beta \in (0,1)$, compute a (tight) prediction region $R$ and a (tight) lower bound $\mu \in (0,1)$ on the containment probability, such that $\satprob(R) \geq \mu$ holds with a confidence level of at least $\beta$.
\end{mdframed}

\noindent
The problem in \cref{sec:Introduction} is a special case of the formal problem, with $\Phi$ the reachability probability over a set of horizons.
In that case, we can overapproximate a prediction region as a rectangle, yielding an interval $[\munderbar{c}, \bar{c}]$ for every horizon $t$ that defines where the two step functions (see below \cref{def:probabilityCurve}) change. 
We smoothen these step functions (similar to probability curves) to obtain the following definition:

\begin{definition}[Prediction region for a probability curve]
    \label{def:ConfidenceProbCurve}
    A \emph{prediction region} $R$ over a probability curve $\probc$ is given by two curves $\munderbar{c}, \bar{c}: \Q_{\geq0} \to \R$ as the area in-between: $R = \{(t, y) \in \Q \times \R \mid \munderbar{c}(t) \leq y \leq \bar{c}(t)\}$.
\end{definition}

\noindent
We solve the problem by sampling a finite set $\sampleset$ of parameter values of the \gls{upCTMC} and computing the corresponding solution vectors $\sol(\sampleset)$.
In \cref{sec:Scenario_precise}, we solve the problem assuming that we can compute solution vectors exactly.
In \cref{sec:Scenario_imprecise}, we consider a less restricted setting in which every solution is imprecise, \ie only known to lie in a certain interval.
\section{Precise Sampling-Based Prediction Regions}
\label{sec:Scenario_precise}

In this section, we use scenario optimization~\cite{DBLP:journals/siamjo/CampiG08, campi2018introduction} to compute a high-confidence lower bound on the containment probability.
First, in \cref{subsec:Scenario_confRegions}, we describe how to compute a prediction region using the solution vectors $\sol(\sampleset)$ for the parameter samples $\sampleset$.
In \cref{subsec:Scenario_bounds}, we clarify how to compute a lower bound on the containment probability with respect to this prediction region.
In \cref{subsec:Scenario_algorithm}, we construct an algorithm based on those results that solves the formal problem.

\subsection{Constructing prediction regions}
\label{subsec:Scenario_confRegions}
We assume that we are given a set of solution vectors $\sol(\sampleset)$ obtained from $n$ parameter samples.
We construct a prediction region $R$ based on these vectors such that we can annotate these regions with a lower bound on the containment probability, as in the problem statement. 
For conciseness, we restrict ourselves to the setting where $R$ is a hyperrectangle in $\R^m$, with $m = |\Phi|$ the number of measures, cf.~\cref{remark:ScenarioProblem_generality} below.
In the following, we represent $R$ using two vectors (points) $\munderbar{x}, \bar{x} \in \R^m$ such that, using pointwise inequalities, $R=\{ x \mid \munderbar{x} \leq x \leq \bar{x} \}$. 
For an example of such a rectangular prediction region, see \cref{fig:confidence_2D}.

As also shown in \cref{fig:confidence_2D}, we do \emph{not} require $R$ to contain all solutions in $\sol(\sampleset)$.
Instead, we have two orthogonal goals: we aim to minimize the size of $R$, while also minimizing the (Manhattan) distance of samples to $R$, measured in their 1-norm. 
Solutions contained in $R$ are assumed to have a distance of zero, while solutions not contained in $R$ are called \emph{relaxed}.
These goals define a \emph{multi-objective problem}, which we solve by weighting the two objectives using a fixed parameter $\rho > 0$, called the \emph{cost of relaxation}, that is used to scale the distance to $R$. 
Then, $\rho \to \infty$ enforces $\sol(\sampleset) \subseteq R$, as in the outer box in \cref{fig:confidence_2D}, while for $\rho \to 0$, $R$ is reduced to a point.
Thus, the cost of relaxation $\rho$ is a tuning parameter that determines the size of the prediction region $R$ and hence the fraction of the solution vectors that is contained in $R$
(see \cite{DBLP:journals/mp/CampiG18,Campi2021scenarioMaking} for details).

We capture the problem described above in the following convex optimization problem $\mathfrak{L}^\rho_\mathcal{U}$.
We define the decision variables $\munderbar{x}, \bar{x} \in \R^m$ to represent the prediction region.
In addition, we define a decision variable $\xi_i \in \R^m_{\geq 0}$ for every sample $i = 1,\ldots,n$ that acts as a slack variable representing the distance to $R$.
\begin{subequations}
\begin{align}
    \label{eq:ScenarioProblem_objective}
    \mathfrak{L}_\mathcal{U}^\rho \, \colon \, &\minimize\quad
    \, \| \bar{x} - \munderbar{x} \|_1 + \rho \sum_{i=1}^{n} \| \xi_i \|_1 
    \\
    \label{eq:ScenarioProblem_constraint}
    &
    \text{subject to}\quad 
        \munderbar{x} - \xi_i   \leq \sol(u_i) \leq \bar{x} + \xi_i \quad
    \forall i = 1,\ldots,n.
\end{align}
\label{eq:ScenarioProblem}
\end{subequations}
The objective function in \cref{eq:ScenarioProblem_objective} minimizes the size of $R$ ---by minimizing the sum of the width of the prediction region in all dimensions--- plus $\rho$ times the distances of the samples to $R$.
We denote the optimal solution to problem $\mathfrak{L}^\rho_\mathcal{U}$ for a given $\rho$ by $R^*_\rho, \xi^*_\rho$, where $R^*_\rho = [\munderbar{x}^*_\rho, \bar{x}^*_\rho]$ for the rectangular case.

\begin{assumption}
    \label{def:Uniqueness}
    The optimal solution $R^*_\rho, \xi^*_\rho$ to $\mathfrak{L}^\rho_\mathcal{U}$ exists and is unique.
\end{assumption}

\noindent
Note that \cref{def:solutionFunction} ensures finite-valued solution vectors, thus guaranteeing the existence of a solution to \cref{eq:ScenarioProblem}.
If the solution is not unique, we apply a suitable tie-break rule that selects one solution of the optimal set (e.g., the solution with a minimum Euclidean norm, see~\cite{DBLP:journals/siamjo/CampiG08}).
The following example shows that values of $\rho$ exist for which such a tie-break rule is necessary to obtain a unique solution. 

\begin{figure}[t!]
\centering
\begin{minipage}[b]{.48\textwidth}
  \centering
\iftikzcompile
      \resizebox{.88\linewidth}{!}{%

\begin{tikzpicture}[xscale=0.9, yscale=0.9]

\node [above, rotate=90] at (-0.4cm,2.1cm) {Measure};

\def\hA{0.3cm}
\def\hB{1.0cm}
\def\hC{1.5cm}
\def\hD{2.1cm}
\def\hE{2.65cm}
\def\hF{3.2cm}

\draw[latex-latex] (0.5cm,0.35cm) -- (0.5cm,0.95cm) node[right, midway] {$\delta_1$};
\draw[latex-latex] (0.5cm,1.05cm) -- (0.5cm,1.45cm) node[right, midway] {$\delta_2$};
\draw[latex-latex] (0.5cm,1.55cm) -- (0.5cm,2.05cm) node[right, midway] {$\delta_3$};
\draw[latex-latex] (0.5cm,2.15cm) -- (0.5cm,2.60cm) node[right, midway] {$\delta_4$};
\draw[latex-latex] (0.5cm,2.70cm) -- (0.5cm,3.15cm) node[right, midway] {$\delta_5$};

\fill[red!20] (1.4cm,\hA) rectangle (1.7cm,\hF);
\fill[MidnightBlue!20] (2.8cm,\hB) rectangle (3.1cm,\hE);
\fill[OliveGreen!20] (4.2cm,\hC) rectangle (4.5cm,\hD);

\draw[dashed,color=red!80] (0cm,\hA) -- (1.7cm,\hA);
\draw[dashed,color=red!80] (0cm,\hF) -- (1.7cm,\hF);

\draw[dashed,color=MidnightBlue!80] (0cm,\hB) -- (3.1cm,\hB);
\draw[dashed,color=MidnightBlue!80] (0cm,\hE) -- (3.1cm,\hE);

\draw[dashed,color=OliveGreen!80] (0cm,\hC) -- (4.5cm,\hC);
\draw[dashed,color=OliveGreen!80] (0cm,\hD) -- (4.5cm,\hD);

\draw (1.55cm, \hF) node[above, align=center] {[$\munderbar{x}^*_\rho, \bar{x}^*_\rho]$ \\ $\rho > 1$};
\draw (2.95cm, \hE) node[above, align=center] {[$\munderbar{x}^*_\rho, \bar{x}^*_\rho]$ \\ $\frac{1}{2} < \rho < 1$};
\draw (4.35cm, \hD) node[above, align=center] {[$\munderbar{x}^*_\rho, \bar{x}^*_\rho]$ \\ $\frac{1}{4} < \rho < \frac{1}{2}$};

\draw[-latex] (0cm,0cm) -- (0cm,4cm) ;
\draw[color=black] (-3pt,0pt) -- (3pt,0pt);
\draw[draw=none] (0pt,0pt) -- (-3pt,0pt) node[left] {0};

\draw (0,\hA)  node[circle, fill=black, inner sep=1.5pt] {} node[left] {\scriptsize A};
\draw (0,\hB)  node[circle, fill=black, inner sep=1.5pt] {} node[left] {\scriptsize B};
\draw (0,\hC)  node[circle, fill=black, inner sep=1.5pt] {} node[left] {\scriptsize C};
\draw (0,\hD)  node[circle, fill=black, inner sep=1.5pt] {} node[left] {\scriptsize D};
\draw (0,\hE)  node[circle, fill=black, inner sep=1.5pt] {} node[left] {\scriptsize E};
\draw (0,\hF)  node[circle, fill=black, inner sep=1.5pt] {} node[left] {\scriptsize F};

\end{tikzpicture}

}
\fi
  \captionof{figure}{The prediction region changes with the cost of relaxation $\rho$.}
  \label{fig:1D_example_rho}
\end{minipage}\hfill
\begin{minipage}[b]{.48\textwidth} 
  \centering
\iftikzcompile
      \pgfdeclarelayer{bg}    
\pgfsetlayers{bg,main}  

\begin{tikzpicture}
\begin{axis}[%
    width=\linewidth,
    height=4.20cm,
    ymin=1.13,
    ymax=1.46,
    xmin=1.745,
    xmax=1.83,
    xlabel={Expected \#tokens cell 1},
    ylabel={Exp. \#tokens cell 2},
    legend entries={$\rho=2.00$,$\rho=0.40$,$\rho=0.15$},
    legend pos=north west,
    legend image post style={xscale=0.88},
    legend style={nodes={scale=0.8, transform shape},
                  /tikz/every even column/.append style={column sep=0.1cm}},
    legend columns=3, 
    x label style={at={(axis description cs:0.5,-0.13)}},
    y label style={yshift=-0.2cm, xshift=-.2cm}
]

\addlegendimage{only marks, mark=square*, color=red!25}
\addlegendimage{only marks, mark=square*, color=MidnightBlue!45}
\addlegendimage{only marks, mark=square*, color=OliveGreen!65}

\addplot[only marks,fill=black,mark size=1.2pt]%
table[] {
x           y
1.791002317	1.287047405
1.760996668	1.178218254
1.801798098	1.373303387
1.793570725	1.229113056
1.771414456	1.256181502
1.756854193	1.318985678
1.764584045	1.352788882
1.813511796	1.348903555
1.77445442	1.33744589
1.776479958	1.258781112
1.801963122	1.337530703
1.787720666	1.341255286
1.763369664	1.285770994
1.797380135	1.34333188
1.822825891	1.293260215
1.79461937	1.283035544
1.817178499	1.285999928
1.798652917	1.315193172
1.776090128	1.200220888
1.791046483	1.333403096
1.802752741	1.252911935
1.759647063	1.247647172
1.806711333	1.316012061
1.77599166	1.303812584
1.785557608	1.282996083
};

\coordinate (low1) at (axis cs:1.756854193, 1.178218254) {};
\coordinate (upp1) at (axis cs:1.822825891, 1.373303387) {};

\coordinate (low2) at (axis cs:1.760996668, 1.229113056) {};
\coordinate (upp2) at (axis cs:1.813511796, 1.348903555) {};

\coordinate (low3) at (axis cs:1.77445442, 1.258781112) {};
\coordinate (upp3) at (axis cs:1.801798098, 1.33744589) {};

\end{axis}

\begin{pgfonlayer}{bg}
    \filldraw[dotted, color=red, fill=red!25] (low1) rectangle (upp1);
    \filldraw[dotted, color=MidnightBlue, fill=MidnightBlue!45] (low2) rectangle (upp2);
    \filldraw[dotted, color=OliveGreen, fill=OliveGreen!65] (low3) rectangle (upp3);
\end{pgfonlayer}

\end{tikzpicture}
\fi
  \captionof{figure}{Prediction regions as boxes, for different costs of relaxations $\rho$.}
  \label{fig:confidence_2D_relaxed}
\end{minipage}
\end{figure}

\begin{example}
    \label{example:IntuitionRho}
    \cref{fig:1D_example_rho} shows a set of solution vectors in one dimension, labeled $A$--$F$.
    Consider prediction region $R_1 = [A,F]$.
    The corresponding objective value \cref{eq:ScenarioProblem_objective} is $\| \bar{x} - \munderbar{x} \| + \rho \cdot \sum \xi_i = \| \bar{x} - \munderbar{x} \| =  \delta_1 + \cdots + \delta_5$,  as all $\xi_i = 0$.
    For prediction region $R_2 = [B, E]$, the objective value is $\delta_2 + \delta_3 + \delta_4 + \rho \cdot \delta_1 + \rho \cdot \delta_5$.
    Thus, for $\rho > 1$, solving $\mathfrak{L}^\rho_\mathcal{U}$ yields $R_1$ whereas for $\rho < 1$, relaxing solutions A and F is cheaper than not doing so, so $R_2$ is optimal.
    When $\rho = 1$, however, relaxing solutions A and F yields the same cost as not relaxing these samples, so a tie-break rule is needed (see above).
    For $\rho < \frac{1}{2}$, relaxing samples A, B, E, and F is cost-optimal, resulting in the prediction region containing exactly $\{C, D\}$. \qed
\end{example}

\noindent
Similarly, we can consider cases with more samples and multiple measures,  
\ifappendix
    as shown in \cref{fig:confidence_2D_relaxed}, which we discuss in more detail in \cref{app:Examples}.
\else
    as shown in \cref{fig:confidence_2D_relaxed} (see \cite[App.~A]{Badings2022extended} for more details).
\fi
The three prediction regions in \cref{fig:confidence_2D_relaxed} are obtained for different costs of relaxation $\rho$.
For $\rho=2$, the region contains all vectors, while for a lower $\rho$, more vectors are left outside.

\begin{remark}
    \label{remark:ScenarioProblem_generality}
    While problem $\mathfrak{L}^\rho_\mathcal{U}$ in \cref{eq:ScenarioProblem} yields a rectangular prediction region, we can also produce other shapes. 
    We may, e.g., construct a Pareto front as in \cref{fig:confidence_pareto}, by adding additional affine constraints~\cite{boyd_convex_optimization}.
    In fact, our only requirement is that the objective function is convex, and the constraints are convex in the decision variables (the dependence of the constraints on $u$ may be arbitrary)~\cite{Campi2021scenarioMaking}. \qed
\end{remark}

\subsection{Bounding the containment probability}
\label{subsec:Scenario_bounds}
The previous section shows how we compute a prediction region based on convex optimization.
We now characterize a valid high-confidence lower bound on the containment probability w.r.t.\ the prediction region given by the optimal solution to this optimization problem.
Toward that result, we introduce the so-called \emph{complexity} of a solution to problem $\mathfrak{L}^\rho_\mathcal{U}$ in \cref{eq:ScenarioProblem}, a concept used in~\cite{Campi2021scenarioMaking} that is related to the compressibility of the solution vectors $\sol(\sampleset)$:

\begin{definition}[Complexity]
    \label{def:Complexity}
    For $\mathfrak{L}^\rho_\mathcal{U}$ with optimal solution $R^*_\rho, \xi^*_\rho$, consider a set $\mathcal{W} \subseteq \sampleset$ and the associated problem $\mathfrak{L}^\rho_\mathcal{W}$ with optimal solution $\tilde{R}_\rho, \tilde{\xi}_\rho$.
    The set $\mathcal{W}$ is \emph{critical}, if
    \[ 
    \tilde{R}_\rho = R^*_\rho
    \quad\text{ and }\quad \{  u_i \mid \xi^*_{\rho,i} > 0 \} \subseteq \mathcal{W}.
    \]
    The \emph{complexity} $c^*_\rho$ of $R^*_\rho, \xi^*_\rho$ is the cardinality of the smallest critical set. 
    We also call  $c^*_\rho$ the complexity of $\mathfrak{L}^\rho_\mathcal{U}$.
\end{definition}

\noindent
If a sample $u_i$ has a value $\xi^*_{\rho,i} > 0$, its solution vector has a positive distance to the prediction region, $R^*_\rho$. (i.e., $[\munderbar{x}^*_\rho, \bar{x}^*_\rho]$ for the rectangular case).
Thus, the complexity $c^*_\rho$ is the number of samples for which $\sol(u_i) \notin
R^*_\rho$,
plus the minimum number of samples needed on the boundary of the region to keep the solution unchanged.
We describe in \cref{subsec:Scenario_algorithm} how we algorithmically determine the complexity.

\begin{example}
    \label{Ex:Complexity}
    In \cref{fig:confidence_2D_relaxed}, the prediction region for $\rho = 2$ contains all solution vectors, so $\xi^*_{2,i} = 0 \,\forall i$.
    Moreover, if we remove \emph{all but four} solutions (the ones on the boundary of the region), the optimal solution to problem $\mathfrak{L}_\mathcal{U}^\rho$ remains unchanged, so the complexity is $c^*_{1.12} = 0 + 4$.
    Similarly, the complexity for $\rho=0.4$ is $c^*_{0.4} = 8+2 = 10$ (8 solutions outside the region, and 2 on the boundary).
    \qed
\end{example}
\noindent Recall that \cref{def:ContainmentProbability} defines the containment probability of a generic prediction region $R$, so $\satprob(R^*_\rho)$ is the containment probability w.r.t. the optimal solution to $\mathfrak{L}_\mathcal{U}^\rho$.
We adapt the following theorem from \cite{Campi2021scenarioMaking}, which gives a lower bound on the containment probability $\satprob(R^*_\rho)$ of an optimal solution to $\mathfrak{L}^\rho_\mathcal{U}$ for a predefined value of $\rho$.
This lower bound is correct with a user-defined confidence level of $\beta \in (0,1)$, which we typically choose close to one (e.g., $\beta = 0.99$).

\begin{theorem}
    \label{thm:RiskAndComplexity}
    Let $\sampleset$ be a set of $n$ samples, and let $c^*$ be the complexity of problem $\mathfrak{L}^\rho_\mathcal{U}$. 
    For any confidence level $\beta \in (0,1)$ and any upper bound $d^* \geq c^*$, it holds~that
    \begin{equation}
        \label{eq:RiskAndComplexity_bound}
        \amsmathbb{P}^n \Big\{ 
            \satprob\big(R^*_\rho\big) \geq \eta(d^*)
        \Big\}
        \geq \beta,
    \end{equation}
    where $R^*_\rho$ is the prediction region for $\mathfrak{L}^\rho_\mathcal{U}$. 
    Moreover, $\eta$ is a function defined as $\eta(n) = 0$, and otherwise, $\eta(c)$ is the smallest positive real-valued solution to the following polynomial equality in the $t$ variable for a complexity of $c$:
    \begin{equation}
        \label{eq:RiskAndComplexity_polynomial}
        \binom nc t^{n-c} - \frac{1-\beta}{2n}\sum_{i=c}^{n-1} \binom ic t^{i-c} - \frac{1-\beta}{6n} \sum_{i=n+1}^{4n} \binom ic t^{i-c} = 0.
    \end{equation}
\end{theorem}

\begin{figure}[t!]
\centering
\begin{subfigure}[b]{0.47\linewidth}
\iftikzcompile
    \resizebox{.95\linewidth}{!}{%

\begin{tikzpicture}
  \begin{axis}[
      width=\linewidth,
      height=3.4cm,
      ymajorgrids,
      grid style={dashed,gray!30},
      xlabel=Complexity ($c$),
      ylabel=Lower bound ($\eta$),
      x label style={at={(axis description cs:0.4,-0.2)}},
      y label style={at={(axis description cs:-0.15,0.39)}},
      x tick label style={rotate=90,anchor=east},
      xmin=0,
      xmax=25,
      ymin=0,
      ymax=1,
      xtick={0,5,...,25},
      every axis plot/.append style={line width=1pt},
      legend cell align={left},
      legend columns=1,
      legend pos=north east,
      legend image post style={xscale=0.4},
      legend style={nodes={scale=0.8, transform shape}}
    ]
    
    \addplot[mark=none, color=MidnightBlue] table[x=k, y=b0.9, col sep=semicolon] {Figures/Tikz/Data/satprob-bounds_N=25.csv};
    
    \addplot[mark=none, color=BurntOrange] table[x=k, y=b0.99, col sep=semicolon] {Figures/Tikz/Data/satprob-bounds_N=25.csv};
    
    \addplot[mark=none, color=OliveGreen] table[x=k, y=b0.999, col sep=semicolon] {Figures/Tikz/Data/satprob-bounds_N=25.csv};
    
    \legend{{$\beta=0.9$},{$\beta=0.99$},{$\beta=0.999$}}
  \end{axis}
\end{tikzpicture}

}
\fi
    \captionof{figure}{Number of samples $n=25$.}
    \label{fig:exampleBounds25}
\end{subfigure}
\hfill
\begin{subfigure}[b]{0.47\linewidth}
\iftikzcompile
    \resizebox{.95\linewidth}{!}{%

\begin{tikzpicture}
  \begin{axis}[
      width=\linewidth,
      height=3.4cm,
      ymajorgrids,
      grid style={dashed,gray!30},
      xlabel=Complexity ($c$),
      ylabel=Lower bound ($\eta$),
      x label style={at={(axis description cs:0.4,-0.2)}},
      y label style={at={(axis description cs:-0.15,0.39)}},
      x tick label style={rotate=90,anchor=east},
      xmin=0,
      xmax=100,
      ymin=0,
      ymax=1,
      every axis plot/.append style={line width=1pt},
      legend cell align={left},
      legend columns=1,
      legend pos=north east,
      legend image post style={xscale=0.40},
      legend style={nodes={scale=0.8, transform shape}}
    ]
    
    \addplot[mark=none, color=MidnightBlue, mark size=1.3pt] table[x=k, y=b0.9, col sep=semicolon] {Figures/Tikz/Data/satprob-bounds_N=100.csv};
    
    \addplot[mark=none, color=BurntOrange, mark size=1.3pt] table[x=k, y=b0.99, col sep=semicolon] {Figures/Tikz/Data/satprob-bounds_N=100.csv};
    
    \addplot[mark=none, color=OliveGreen, mark size=1.3pt] table[x=k, y=b0.999, col sep=semicolon] {Figures/Tikz/Data/satprob-bounds_N=100.csv};
    
    \legend{{$\beta=0.9$},{$\beta=0.99$},{$\beta=0.999$}}
  \end{axis}
\end{tikzpicture}

}
\fi
    \captionof{figure}{Number of samples $n=100$.}
    \label{fig:exampleBounds100}
\end{subfigure}
\caption{Lower bounds $\eta$ on the containment probability as a function of the complexity $c$, obtained from \cref{thm:RiskAndComplexity} for different confidence levels $\beta$.} 
\label{fig:exampleBounds}
\end{figure}
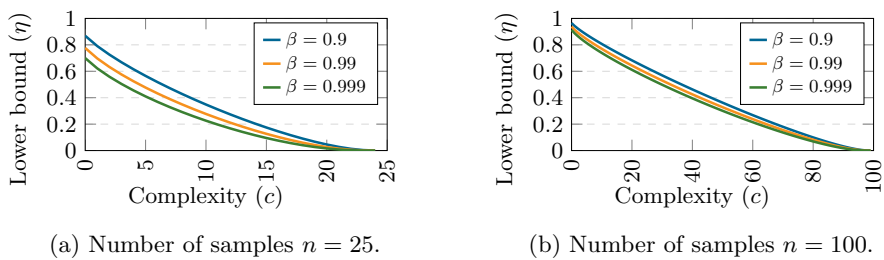

\noindent
We provide the proof of \cref{thm:RiskAndComplexity} in 
\ifappendix
    \cref{app:proof:thm1}.
\else
    \cite[App.~B.1]{Badings2022extended}.
\fi
With a probability of at least $\beta$, \cref{thm:RiskAndComplexity} yields a correct lower bound.
That is, if we solve $\mathfrak{L}^\rho_\mathcal{U}$ for many more sets of $n$ parameter samples (note that, as the samples are \gls{iid}, these sets are drawn according to the product probability $\amsmathbb{P}^n$), the inequality in \cref{eq:RiskAndComplexity_bound} is incorrect for \emph{at most} a $1-\beta$ fraction of the cases.
We plot the lower bound $\eta(c)$ as a function of the complexity $c = 0, \ldots, n$ in \cref{fig:exampleBounds}, for different samples sizes $n$ and confidence levels $\beta$.
These figures show that an increased complexity leads to a lower $\eta$, while increasing the sample size leads to a tighter bound.

\begin{example}
    \label{Ex:LowerBounds}
    We continue \cref{Ex:Complexity}.
    Recall that the complexity for the outer region in \cref{fig:confidence_2D_relaxed} is $c^*_{1.12} = 4$.
    With \cref{thm:RiskAndComplexity}, we compute that, for a confidence level of $\beta = 0.9$, the containment probability for this prediction region is at least $\eta = 0.615$ (cf.~\cref{fig:exampleBounds25}).
    For a stronger confidence level of $\beta=0.999$, we obtain a more conservative lower bound of $\eta = 0.455$.
    \qed
\end{example}

\subsection{An algorithm for computing prediction regions}
\label{subsec:Scenario_algorithm}

We combine the previous results in our algorithm, which is outlined in \cref{fig:Approach}. 
The goal is to obtain a set of prediction regions as in \cref{fig:confidence_2D_relaxed} and their associated lower bounds.
To strictly solve the problem statement, assume $k=1$ in the exposition below. 
We first outline the complete procedure before detailing Steps 4 and 5.

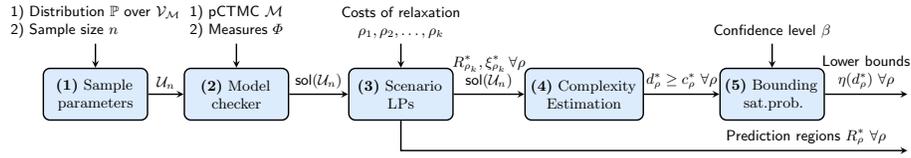
\begin{figure}[t!]
	\centering
\iftikzcompile
	\usetikzlibrary{calc}
\usetikzlibrary{arrows.meta}
\usetikzlibrary{positioning}
\usetikzlibrary{fit}
\usetikzlibrary{shapes}

\tikzstyle{node} = [rectangle, rounded corners, minimum width=2cm, text width=2cm, minimum height=1.1cm, text centered, draw=black, fill=plotblue!15]
\tikzstyle{nodebig} = [rectangle, rounded corners, minimum width=2.3cm, text width=2.3cm, minimum height=1.1cm, text centered, draw=black, fill=plotblue!15]

\resizebox{\linewidth}{!}{%
	\begin{tikzpicture}[node distance=4.2cm,->,>=stealth,line width=0.3mm,auto,
		main node/.style={circle,draw,font=\sffamily\bfseries}]
		
		\newcommand\xshift{-1cm}
		
		\newcommand\yshift{-3.0cm}
		
		\node (sampler) [node] {\textbf{(1)} Sample \\ parameters};
		\node (checker) [node, right of=sampler, xshift=-1.2cm] {\textbf{(2)} Model \\ checker};
		\node (LP) [node, right of=checker, xshift=-0.7cm] {\textbf{(3)} Scenario \\ LPs};
		\node (Comp) [nodebig, right of=LP, xshift=-0.3cm] {\textbf{(4)} Complexity \\ Estimation};
		\node (bound) [node, right of=Comp, xshift=-0.2cm] {\textbf{(5)} Bounding \\ sat.prob.}; 
		
		
		\node (in1) [above of=sampler, yshift=\yshift] {};
		\node (in2) [above of=checker, yshift=\yshift] {};
		\node (in3) [above of=LP, yshift=\yshift] {};
		\node (in4) [above of=bound, yshift=\yshift] {};
		\node (out) [right of=bound, xshift=-1.2cm] {};
		\node (out2) [below of=out, yshift=3.0cm] {};
		
		
		\draw [->] (in1) -- (sampler) node [pos=0.0, above, align=left, font=\sffamily\small] {1) Distribution $\PP$ over $\paramspace$ \\ 2) Sample size $n$};
		
		\draw [->] (in2) -- (checker) node [pos=0.0, above, align=left, font=\sffamily\small] {1) \gls{pCTMC} $\pctmc$ \\ 2) Measures $\Phi$};
		
		\draw [->] (in3) -- (LP) node [pos=0.0, above, align=center, font=\sffamily\small] {Costs of relaxation \\ $\rho_1, \rho_2, \ldots, \rho_k$};
		
		\draw [->] (in4) -- (bound) node [pos=0.0, above, align=center, font=\sffamily\small] {Confidence level $\beta$};
		
		
		\draw [->] (sampler) -- (checker) node [pos=0.5, above, align=center, font=\sffamily\small] {$\sampleset$};
		
		\draw [->] (checker) -- (LP) node [pos=0.5, above, align=center, font=\sffamily\small] {$\sol(\sampleset)$};
		
		\draw [->] (LP) -- (Comp) node [pos=0.5, above, align=center, font=\sffamily\small] {$R^*_{\rho_k}, \xi^*_{\rho_k} \, \forall \rho$ \\ $\sol(\sampleset)$};
		
		\draw [->] (Comp) -- (bound) node [pos=0.5, above, align=center, font=\sffamily\small] {$d^*_{\rho} \geq c^*_{\rho} \,\,\forall \rho$}; 
		
		\draw [->] (bound) -- (out) node [pos=0.5, above, align=center, font=\sffamily\small] {Lower bounds \\ $\eta(d^*_\rho) \,\, \forall \rho$};
		
		\draw [->] (LP) |- (out2) node [pos=0.9, above, align=center, font=\sffamily\small] {Prediction regions $R^*_\rho \,\, \forall \rho$};

	\end{tikzpicture}
}
\fi
	\vspace{-2em}
	\caption{
		Overview of our approach for solving the problem statement.
	}
	\label{fig:Approach}
\end{figure}

As preprocessing steps, given a upCTMC $(\pctmc, \PP)$, we first (1)~sample a set $\sampleset$ of $n$ parameter values.
Using $\pctmc$ and $\Phi$, a (2)~model checking algorithm then computes the solution vector $\sol^\Phi_{\pctmc}(u)$ for each $u \in \sampleset$, yielding the set of solutions $\sol(\sampleset)$.
We then use $\sol(\sampleset)$ as basis for (3)~the scenario problem $\mathfrak{L}_\mathcal{U}^\rho$ in \cref{eq:ScenarioProblem}, which we solve for $k$ predefined values $\rho_1, \ldots, \rho_k$, yielding $k$ prediction regions $R^*_{\rho_1}, \ldots R^*_{\rho_k}$.
We (4)~compute an upper bound $d^*_\rho$ on the complexity $c^*_\rho \, \forall\rho$.
Finally, we (5)~use the result in \cref{thm:RiskAndComplexity}, for a given confidence $\beta$, to compute the lower bound on the containment probability $\eta(d^*_\rho)$ of $R^*_\rho$.
Using \cref{def:ConfidenceProbCurve}, we can postprocess this region to a prediction region over the probability curves.

\smallpar{Step (3): Choosing values for $\rho$.}
\cref{example:IntuitionRho} shows that relaxation of additional solution vectors (and thus a change in the prediction region) only occurs at \emph{critical} values of $\rho = \frac{1}{n}$, for $n \in \N$.
In practice, we will use $\rho = \frac{1}{n+0.5}$ for $\pm 10$ values of $n \in \N$ to obtain gradients of prediction regions as in \cref{sec:Experiments}.

\smallpar{Step (4): Computing complexity.}
Computing the complexity $c^*_\rho$ is a combinatorial problem in general~\cite{garatti2021risk}, because we must consider the removal of all combinations of the solutions on the boundary of the prediction region $R^*_\rho$.
In practice, we compute an upper bound $d^*_\rho \geq c^*_\rho$ on the complexity via a greedy algorithm.
Specifically, we iteratively solve $\mathfrak{L}^\rho_\mathcal{U}$ in \cref{eq:ScenarioProblem} with \emph{one more sample on the boundary removed}.
If the optimal solution is unchanged, we conclude that this sample does not contribute to the complexity.
If the optimal solution is changed, we put the sample back and proceed by removing a different sample.
This greedy algorithm terminates when we have tried removing all solutions on the boundary.

\smallpar{Step (5): Computing lower bounds.} \cref{thm:RiskAndComplexity} characterizes a computable function $B(d^*, n, \beta)$ that returns zero for $d^*=n$ (i.e., all samples are critical), and otherwise uses the polynomial \cref{eq:RiskAndComplexity_polynomial} to obtain $\eta$, which we solve with an approximate root finding method in practice (see~\cite{garatti2019risk} for details on how to ensure that we find the smallest root). 
For every upper bound on the complexity $d^*$ and any requested confidence, we obtain the lower bound $\eta = B(d^*, n, \beta)$ for the containment probability w.r.t.\ the prediction region $R^*_\rho$.
\section{Imprecise Sampling-Based Prediction Regions}
\label{sec:Scenario_imprecise}

Thus far, we have solved our problem statement under the assumption that we compute the solution vectors precisely (up to numerics).
For some models, however, computing precise solutions is expensive.
In such a case, we may choose to compute an approximation, given as an \emph{interval} on each entry of the solution function.
In this section, we deal with such \emph{imprecise solutions}.

\smallpar{Setting.}
Formally, imprecise solutions are described by the bounds $\sol^-(u), \sol^+(u) \in \R^m$ such that $\sol^-(u) \leq \sol(u) \leq \sol^+(u)$ holds with pointwise inequalities.
Our goal is to compute a prediction region $R$ and a (high-confidence) lower bound $\mu$ such that $\satprob(R) \geq \mu$, i.e., a lower bound on the probability that any \emph{precise solution} $\sol(u)$ is contained in $R$.
However, we must now compute $R$ and $\satprob(R)$ from the imprecise solutions $\sol^-, \sol^+$.
Thus, we aim to provide a guarantee with respect to the \emph{precise} solution $\sol(u)$, based on \emph{imprecise} solutions.

\smallpar{Challenge.}
Intuitively, if we increase the (unknown) prediction region $R^*$ from problem $\mathfrak{L}_\mathcal{U}^\rho$ (for the unknown precise solutions) while also overapproximating the complexity of $\mathfrak{L}_\mathcal{U}^\rho$, we obtain sound bounds.
We formalize this idea as follows.
\begin{lemma}
    \label{lemma:Implication}
    Let $R^*_\rho$ be the prediction region and $c^*_\rho$ the complexity that result from solving $\mathfrak{L}^\rho_\mathcal{U}$ for the precise (unknown) solutions $\sol(\sampleset)$.
    Given a set $R \in \R^n$ and $d \in \N$, for any confidence level $\beta \in (0,1)$, the following implication holds:
    \begin{equation}
        \label{eq:ImpreciseSolutions_implication}
        R^*_\rho \subseteq R
        \text{ and }
        c^*_\rho \leq d
        \enskip\implies\enskip
        \amsmathbb{P}^n \Big\{ 
            \satprob\big(R\big) \geq \eta(d)
        \Big\} \geq \beta,
    \end{equation}
    where $\eta(n) = 0$, and otherwise, $\eta(d)$ is the smallest positive real-valued solution to the polynomial equality in \cref{eq:RiskAndComplexity_polynomial}.
\end{lemma}

\noindent
The proof is in 
\ifappendix
    \cref{app:proof:implication}.
\else
    \cite[App.~B.2]{Badings2022extended}.
\fi
In what follows, we clarify how we compute the appropriate $R$ and $d$ in \cref{lemma:Implication}. 
As we will see, in contrast to \cref{sec:Scenario_precise}, these results do \emph{not} carry over to other definitions $\mathfrak{L}_\mathcal{U}^{\rho}$ (for non-rectangular regions $R$).

\subsection{Prediction regions on imprecise solutions}
In this section, we show how to compute $R \supseteq R^*_{\rho}$, satisfying the first term in the premise of \cref{lemma:Implication}. 
We construct a \emph{conservative box} around the imprecise solutions as in \cref{fig:imprecise_2D}, containing both $\sol^-(u)$ and $\sol^+(u)$.
We compute this box by solving the following problem $\mathfrak{G}_\mathcal{U}^{\rho}$ as a modified version of $\mathfrak{L}^\rho_\mathcal{U}$ in~\cref{eq:ScenarioProblem}: 
\begin{subequations}
\begin{align}
    \label{eq:ScenarioProblem_imp_objective}
    \mathfrak{G}_\mathcal{U}^{\rho} \, \colon \, &\minimize
    \, \| \bar{x} - \munderbar{x} \|_1 + \rho \sum_{i=1}^{n} \| \xi_i \|_1 
    \\
    \label{eq:ScenarioProblem_imp_constraint}
    &
    \text{subject to}\quad 
        \munderbar{x} - \xi_i \leq \sol^-(u_i), \enskip
        \sol^+(u_i) \leq \bar{x} + \xi_i \quad
    \forall i = 1, \ldots, n.
\end{align}
\label{eq:ScenarioProblem_imp}
\end{subequations}
We denote the optimal solution of $\mathfrak{G}_\mathcal{U}^{\rho}$ by $[\munderbar{x}'_\rho, \bar{x}'_\rho], \xi'_\rho$ (recall that the optimum to $\mathfrak{L}_\mathcal{U}^\rho$ is written as $[\munderbar{x}^*_\rho, \bar{x}^*_\rho], \xi^*_\rho$).\footnote{We write $[\munderbar{x}^*_\rho, \bar{x}^*_\rho]$ and $[\munderbar{x}'_\rho, \bar{x}'_\rho]$, as results in \cref{sec:Scenario_imprecise} apply only to rectangular regions.}
If a sample $u_i \in \paramspace$ in problem $\mathfrak{G}_\mathcal{U}^{\rho}$ is relaxed (i.e., has a non-zero $\xi_i$), part of the interval $[\sol^-(u_i), \sol^+(u_i)]$ is not contained in the prediction region.
The following result (for which the proof is in 
\ifappendix
    \cref{app:proof:subsetrelation})
\else
    \cite[App.~B.3]{Badings2022extended}.
\fi
relates $\mathfrak{L}_\mathcal{U}^\rho$ and $\mathfrak{G}^\rho_\mathcal{U}$, showing that we can use $[\munderbar{x}'_\rho, \bar{x}'_\rho]$ as $R$ in~\cref{lemma:Implication}.

\begin{theorem}
    \label{thm:subset_relation}
    Given $\rho$, sample set $\sampleset$, and prediction region $[\munderbar{x}'_\rho, \bar{x}'_\rho]$ to problem $\mathfrak{G}^\rho_\mathcal{U}$, it holds that $[\munderbar{x}^*_\rho, \bar{x}^*_\rho] \subseteq [\munderbar{x}'_\rho, \bar{x}'_\rho]$, with $[\munderbar{x}^*_\rho, \bar{x}^*_\rho]$ the optimal solution to $\mathfrak{L}^\rho_\mathcal{U}$.
\end{theorem}

\noindent
We note that this result is not trivial.
In particular, the entries $\xi_i$ from both LPs are incomparable, as are their objective functions. 
Instead, \cref{thm:subset_relation} relies on two observations. 
First, due to the use of the 1-norm, the LP $\mathfrak{G}^\rho_\mathcal{U}$ can be decomposed into $n$ individual LPs, whose results combine into a solution to the original LP. 
This allows us to consider individual dimensions. 
Second, the solution vectors that are relaxed depend on the value of $\rho$ and on their \emph{relative order}, but not on the \emph{precise position} within that order, which is also illustrated by \Cref{example:IntuitionRho}.
In combination with the observation from \cref{example:IntuitionRho} that the \emph{outermost} samples are relaxed at the (relatively) highest $\rho$, we can provide conservative guarantees on which samples are (or are surely not) relaxed.
We formalize these observations and provide a proof of \cref{thm:subset_relation} in 
\ifappendix
    \Cref{app:proof:subsetrelation}.
\else
    \cite[App.~B.3]{Badings2022extended}.
\fi

\begin{figure}[t!]
\centering
\begin{minipage}[b]{.48\textwidth}
  \centering
\iftikzcompile
      \resizebox{.75\linewidth}{!}{%
\begin{tikzpicture}[yscale=1.16]

\def\Z{0cm}
\def\hZ{0.3cm}

\def\A{0.2cm}
\def\hA{0.6cm}

\def\B{0.4cm}
\def\hB{1.0cm}

\def\C{0.8cm}
\def\hC{1.4cm}

\def\D{1.2cm}
\def\hD{1.8cm}

\def\E{2.0cm}
\def\hE{2.3cm}

\fill[MidnightBlue!20] (-0.3cm,\Z) rectangle (0cm,\hZ);
\fill[MidnightBlue!20] (0cm,\A) rectangle (0.3cm,\hA);
\fill[MidnightBlue!20] (-0.3cm,\B) rectangle (0cm,\hB);
\fill[MidnightBlue!20] (0cm,\C) rectangle (0.3cm,\hC);
\fill[MidnightBlue!20] (-0.3cm,\D) rectangle (0cm,\hD);
\fill[MidnightBlue!20] (0cm,\E) rectangle (0.3cm,\hE);

\draw[latex-] (0.3cm,\hE) -- (0.6cm,\hE) node[right, align=left] 
{${\sol}^+(u_1)$};
\draw[latex-] (-0.3cm,\hD) -- (-0.6cm,\hD) node[left, align=right] 
{${\sol}^+(u_2)$};
\draw[latex-] (0.3cm,\hC) -- (0.6cm,\hC) node[right, align=left] 
{${\sol}^+(u_3)$};
\draw[latex-] (-0.3cm,\hB) -- (-0.6cm,\hB) node[left, align=right] 
{${\sol}^+(u_4)$};
\draw[latex-] (0.3cm,\hA) -- (0.6cm,\hA) node[right, align=left] 
{${\sol}^+(u_5)$};
\draw[latex-] (-0.3cm,\hZ) -- (-0.6cm,\hZ) node[left, align=right] 
{${\sol}^+(u_6)$};

\draw[dashed,color=red!80] (0cm,\hB) -- (1.6cm,\hB) node[right] {$\bar{x}^+_\rho$};

\draw[-latex] (0cm,0.2cm) -- (0cm,2.5cm) node[pos=.055, below] {$\vdots$};

\end{tikzpicture}
}
\fi
  \captionof{figure}{Imprecise solutions and the upper bound $\bar{x}'_\rho$ of the prediction region.}
  \label{fig:1D_example_imprecise}
\end{minipage}
\hfill
\begin{minipage}[b]{.48\textwidth}
  \centering
\iftikzcompile
      \resizebox{.95\linewidth}{!}{%

\begin{tikzpicture}
\begin{axis}[%
    width=\linewidth,
    height=4.3cm,
    ymin=-0.05,
    ymax=1.1,
    xmin=0,
    xmax=1,
    xlabel={Reliability($t_1$)},
    ylabel={Expected cost($t_1$)},
    legend entries={Imprecise $(c^+_\rho = 3)$,Precise $(c^*_\rho = 4)$},
    legend pos=north west,
    legend image post style={xscale=0.55},
    legend style={
        nodes={scale=0.88, transform shape},
        row sep=-2pt
        },
    x label style={at={(axis description cs:0.5,-0.12)}},
    ytick={0, 0.2, 0.4, 0.6, 0.8, 1.0},
]

\addlegendimage{dotted, very thick, color=red!60}
\addlegendimage{dashed, very thick, color=MidnightBlue!60}

\coordinate (u1) at (axis cs:0.13, 0.13) {};
\coordinate (u2) at (axis cs:0.65, 0.50) {};
\coordinate (u3) at (axis cs:0.73, 0.26) {};
\coordinate (u4) at (axis cs:0.38, 0.40) {};
\coordinate (u5) at (axis cs:0.20, 0.55) {};
\coordinate (u6) at (axis cs:0.53, 0.30) {};
\coordinate (u7) at (axis cs:0.91, 0.56) {};

\coordinate (low1) at (axis cs:0.09, 0.08) {};
\coordinate (upp1) at (axis cs:0.31, 0.22) {};

\coordinate (low2) at (axis cs:0.62, 0.45) {};
\coordinate (upp2) at (axis cs:0.83, 0.65) {};

\coordinate (low3) at (axis cs:0.69, 0.200924) {};
\coordinate (upp3) at (axis cs:0.77, 0.38908) {};

\coordinate (low4) at (axis cs:0.33, 0.29) {};
\coordinate (upp4) at (axis cs:0.45, 0.47) {};

\coordinate (low5) at (axis cs:0.16, 0.35) {};
\coordinate (upp5) at (axis cs:0.30, 0.60) {};

\coordinate (low6) at (axis cs:0.48, 0.23) {};
\coordinate (upp6) at (axis cs:0.57, 0.37) {};

\coordinate (low7) at (axis cs:0.86, 0.46) {};
\coordinate (upp7) at (axis cs:0.95, 0.68) {};

\coordinate (low_box) at (axis cs:0.09, 0.08) {};
\coordinate (upp_box) at (axis cs:0.83, 0.65) {};

\coordinate (low_box2) at (axis cs:0.13, 0.13) {};
\coordinate (upp_box2) at (axis cs:0.73, 0.55) {};

\end{axis}

\node at (u1) {\textbullet};
\node at (u2) {\textbullet};
\node at (u3) {\textbullet};
\node at (u4) {\textbullet};
\node at (u5) {\textbullet};
\node at (u6) {\textbullet};
\node at (u7) {\textbullet};

\filldraw[densely dotted, color=gray, fill=none, thick] (low1) rectangle (upp1);
\filldraw[densely dotted, color=gray, fill=none, thick] (low2) rectangle (upp2);
\filldraw[densely dotted, color=gray, fill=none, thick] (low3) rectangle (upp3);
\filldraw[densely dotted, color=gray, fill=none, thick] (low4) rectangle (upp4);
\filldraw[densely dotted, color=gray, fill=none, thick] (low5) rectangle (upp5);
\filldraw[densely dotted, color=gray, fill=none, thick] (low6) rectangle (upp6);
\filldraw[densely dotted, color=gray, fill=none, thick] (low7) rectangle (upp7);

\filldraw[dotted, color=red!60, fill=none, very thick] (low_box) rectangle (upp_box);
\filldraw[dashed, color=MidnightBlue!60, fill=none, very thick] (low_box2) rectangle (upp_box2);

\node at (u1) [above right, yshift=-0.08cm] {1};
\node at (u2) [right] {2};
\node at (u3) [right] {3};
\node at (u4) [below] {4};
\node at (u5) [below] {5};
\node at (u6) [above] {6};
\node at (u7) [above] {7};

\end{tikzpicture}

}
\fi
  \captionof{figure}{Complexity of the imprecise solution vs. that of the precise solution.}
  \label{fig:imprecise_2D}
\end{minipage}
\end{figure}

\subsection{Computing the complexity}
To satisfy the second term of the premise in \cref{lemma:Implication}, we compute an upper bound on the complexity.
We first present a negative result.
Let the complexity $c'_\rho$ of problem $\mathfrak{G}_\mathcal{U}^{\rho}$ be defined analogous to \cref{def:Complexity}, but with $[\munderbar{x}'_\rho, \bar{x}'_\rho]$ as the region. 
\begin{lemma}
    \label{lemma:Complexity}
    In general, $c^*_\rho \leq c'_\rho$ does not hold.
\end{lemma}
\begin{proof}
    In \cref{fig:imprecise_2D}, the smallest critical set for the imprecise solutions are those labeled $\{1, 2, 7\}$, while this set is $\{1,3,5,7\}$ under precise solutions, so $c^*_\rho > c'_\rho$.
\end{proof}

\noindent
Thus, we cannot upper bound the complexity directly from the result to $\mathfrak{G}_\mathcal{U}^{\rho}$.
We can, however, determine the samples that are certainly \emph{not} in any critical set (recall \cref{def:Complexity}).
Intuitively, a sample is \emph{surely noncritical} if its (imprecise) solution is strictly within the prediction region and does not overlap with any solution on the region's boundary.
In \cref{fig:1D_example_imprecise}, sample $u_6$ is surely noncritical, but sample $u_5$ is not (whether $u_5$ is critical depends on its precise solution).
Formally, let $\delta R$ be the boundary\footnote{The boundary of a compact set is defined as its closure minus its interior~\cite{mendelson1990introduction}.} of region $[\munderbar{x}'_\rho, \bar{x}'_\rho]$, and let $\mathcal{B}$ be the set of samples whose solutions overlap with $\delta R$, which is $\mathcal{B} = \{ u \in \sampleset \,\colon\, [\sol^-(u), \sol^+(u)] \cap \delta R \neq \varnothing \}$. 
\begin{definition}
    \label{def:SurelyNoncritical}
    For a region $[\munderbar{x}'_\rho, \bar{x}'_\rho]$, let $\mathcal{I} \subset [\munderbar{x}'_\rho, \bar{x}'_\rho]$ be the rectangle of largest volume, such that $\mathcal{I} \cap [\sol^-(u), \sol^+(u)] = \varnothing$ for any $u \in \mathcal{B}$.
    A sample $u_i \in \paramspace$ is \emph{surely noncritical} if $[\sol^-(u_i), \sol^+(u_i)] \subseteq \mathcal{I}$.
    The set of all surely noncritical samples w.r.t. the (unknown) prediction region $[\munderbar{x}^*_\rho, \bar{x}^*_\rho]$ is denoted by $\mathcal{X} \subset \sampleset$.
\end{definition}
\noindent
As a worst case, any sample not surely noncritical can be in the smallest critical set, leading to the following bound on the complexity as required by \cref{lemma:Implication}.
\begin{theorem}
    \label{thm:lower_bound_complexity}
    Let $\mathcal{X}$ be the set of surely noncritical samples. Then $c^*_\rho \leq \vert \sampleset \setminus \mathcal{X} \vert$.
\end{theorem}
\noindent 
The proof is in 
\ifappendix
    \Cref{app:proof:lower_bound_complexity}.
\else
    \cite[App.~B.4]{Badings2022extended}.
\fi
For imprecise solutions, the bound in \cref{thm:lower_bound_complexity} is conservative but can potentially be improved, as discussed in the following.

\subsection{Solution refinement scheme}
\label{subsec:Refinement}
Often, we can \emph{refine} imprecise solutions arbitrarily (at the cost of an increased computation time).
Doing so, we can improve the prediction regions and upper bound on the complexity, which in turn improves the computed bound on the containment probability.
Specifically, we propose the following rule for refining solutions.
After solving $\mathfrak{G}_\mathcal{U}^\rho$ for a given set of imprecise solutions, we refine the solutions on the boundary of the obtained prediction region.
We then resolve problem $\mathfrak{G}_\mathcal{U}^\rho$, thus adding a loop back from (4) to (2) in our algorithm shown in \cref{fig:Approach}.
In our experiments, we demonstrate that with this refinement scheme, we iteratively improve our upper bound $d \geq c^*_\rho$ and the smallest superset $R \supseteq R^*_\rho$.
\section{Batch Verification for CTMCs}
\label{sec:Implementation}

One bottleneck in our method is to obtain the necessary number of solution vectors $\sol(\sampleset)$ by model checking.
The following improvements, while mild, are essential in our implementation and therefore deserve a brief discussion.

In general, computing $\sol(u)$ via model checking consists of two parts.
First, the high-level representation of the \gls{upCTMC} ---given in Prism~\cite{DBLP:conf/cav/KwiatkowskaNP11}, JANI~\cite{DBLP:conf/tacas/BuddeDHHJT17}, or a dynamic fault tree\footnote{Fault trees are a common formalism in reliability engineering~\cite{DBLP:journals/csr/RuijtersS15}.}--- is translated into a concrete CTMC $\instctmc{u}$. Then, from $\instctmc{u}$ we construct $\sol(u)$ using off-the-shelf algorithms~\cite{DBLP:journals/tse/BaierHHK03}.
We adapt the pipeline by tailoring the translation and the approximate analysis as outlined below.

Our implementation supports two methods for building the concrete \gls{CTMC} for a parameter sample:
(1)~by first instantiating the valuation in the specification and then building the resulting concrete \gls{CTMC}, or
(2)~by first building the \gls{pCTMC} $\pctmc$ (only once) and then instantiating it for each parameter sample to obtain the concrete \gls{CTMC}~$\instctmc{u}$.
Which method is faster depends on the specific model (we only report results for the fastest method in \cref{sec:Experiments} for brevity).

\smallpar{Partial models.}
To accelerate the time-consuming computation of solution vectors by model-checking on large models, it is natural to abstract the models into smaller models amenable to faster computations.
Similar to ideas used for dynamic fault trees~\cite{DBLP:journals/tii/VolkJK18} and infinite \glspl{CTMC}~\cite{DBLP:conf/vmcai/RobertsNBMZ22}, we employ an abstraction which only keeps the most relevant parts of a model, \ie, states with a sufficiently large probability to be reached from the initial state(s).
Analysis on this partial model then yields best- and worst-case results for each measure by assuming that all removed states are either target states (best case) or are not (worst case), respectively.
This method returns imprecise solution vectors as used in \cref{sec:Scenario_imprecise}, which can be refined up to an arbitrary precision by retaining more states of the original model.

Similar to building the complete models, two approaches are possible to create the partial models:
(1)~fixing the valuation and directly abstracting the concrete \gls{CTMC}, or
(2)~first building the complete \gls{pCTMC} and then abstracting the concrete \gls{CTMC}.
We reuse partial models for similar valuations to avoid costly computations.
We cluster parameter valuations which are close to each other (in Euclidean distance).
For parameter valuations within one cluster, we reuse the same partial model (in terms of the states), albeit instantiating it according to the precise valuation.
\section{Experiments}
\label{sec:Experiments}
We answer three questions about (a prototype implementation of) our approach:
\begin{compactenum}[\bfseries Q1.]
    \item Can we verify \glspl{CTMC} taking into account the uncertainty about the rates?
    \item How well does our approach scale w.r.t. the number of measures and samples?
    \item How does our approach compare to na\"ive baselines (to be defined below)?
\end{compactenum}

\smallpar{Setup.}
We implement our approach using the explicit engine of \storm~\cite{StormSTTT} and the improvements of Sec.~\ref{sec:Implementation} to sample from \glspl{upCTMC} in Python.
Our current implementation is limited to \gls{pCTMC} instantiations that are \emph{graph-preserving}, i.e. for any pair $s, s' \in S$ either $\rateFun(s, s')[u] = 0$ or $\rateFun(s, s')[u] > 0$ for all $u$.
We solve optimization problems using the ECOS solver~\cite{DBLP:conf/eucc/DomahidiCB13}.
All experiments ran single-threaded on a computer with 32 3.7 GHz cores and 64 GB RAM.
We show the effectiveness of our method on a large number of publicly available \gls{pCTMC}~\cite{DBLP:conf/tacas/HartmannsKPQR19} and fault tree benchmarks~\cite{ruijters2019ffort} across domains 
\ifappendix
    (details in \cref{sec:BenchmarkDetails}).
\else
    (details in \cite[App.~C]{Badings2022extended}).
\fi

{
\setlength{\tabcolsep}{3pt}
\begin{table*}[t!]

\centering
\caption{Excerpt of the benchmark statistics (sampling time is per $100$ \glspl{CTMC}).}

\begin{threeparttable}
\scalebox{0.90}{
\begin{tabular}{lrrrrrrrr}
	\hline
 \rule{0pt}{1.5ex} & & \multicolumn{3}{c}{{Model size}}
  & \multicolumn{2}{c}{{\storm run time [s]}} & \multicolumn{2}{c}{{Scen.opt. time [s]}} \\
  \cmidrule(lr){3-5} \cmidrule(lr){6-7} \cmidrule(lr){8-9}
  benchmark
  & $\vert\Phi\vert$
  & \#pars
  & \#states
  & \#trans 
  & Init.
  & Sample ($\times100$)
  & $N=100$
  & $N=200$
  \\
  \hline
SIR (140)   & 26    & 2     & 9\,996        & 19\,716       & 0.29      & 2947.29   & 18.26     & 63.27  \\
SIR (140)\tnote{a} & 26 & 2 & 9\,996        & 19\,716       & 0.29      & 544.27    & 25.11     & 129.66  \\
Kanban (3)  & 4     & 13    & 58\,400       & 446\,400      & 4.42      & 46.95     & 2.28      & 6.69   \\
Kanban (5)  & 4     & 13    & 2\,546\,432   & 24\,460\,016  & 253.39    & 4363.63   & 2.03      & 5.94    \\
polling (9) & 2     & 2     & 6\,912        & 36\,864       & 0.64      & 22.92     & 2.13      & 6.66     \\
buffer      & 2     & 6     & 5\,632        & 21\,968       & 0.48      & 20.70     & 1.21      & 4.15      \\
tandem (31) & 2     & 5     & 2\,016        & 6\,819        & 0.11      & 862.41    & 5.19      & 24.30     \\
rbc         & 40    & 6     & 2\,269        & 12\,930       & 0.01      & 1.40      & 5.27      & 16.88     \\
rc (1,1)    & 25    & 21    & 8\,401        & 49\,446       & 27.20     & 74.90     & 5.75      & 20.34     \\
rc (1,1)\tnote{a}   & 25    & 21    & n/a\tnote{b}   & n/a\tnote{b}     & 0.02      & 2.35      & 29.23     & 150.61    \\
rc (2,2)\tnote{a}   & 25    & 29    & n/a\tnote{b}   & n/a\tnote{b}     & 0.03      & 27.77     & 24.86     & 132.63    \\
hecs (2,1)\tnote{a} & 25    & 5     & n/a\tnote{b} & n/a\tnote{b} & 0.02      & 9.83      & 26.78     & 145.77    \\
hecs (2,2)\tnote{a} & 25    & 24    & n/a\tnote{b} & n/a\tnote{b} & 0.02      & 194.25    & 33.06     & 184.32    \\
\hline
\end{tabular}
}
\begin{tablenotes}
        \raggedright
        \item[a]Computed using approximate model checking up to a relative gap between upper\\bound $\sol^+(u)$ and lower bound $\sol^-(u)$ below $1\%$ for every sample $u \in \paramspace$.
        \item[b]Model size is unknown, as the approximation does not build the full state-space.
\end{tablenotes}
\end{threeparttable}

\label{tab:model_information}
\end{table*}
}

\begin{figure}[t!]
\begin{minipage}{.48\textwidth}
  \centering
  \includegraphics[width=.85\linewidth]{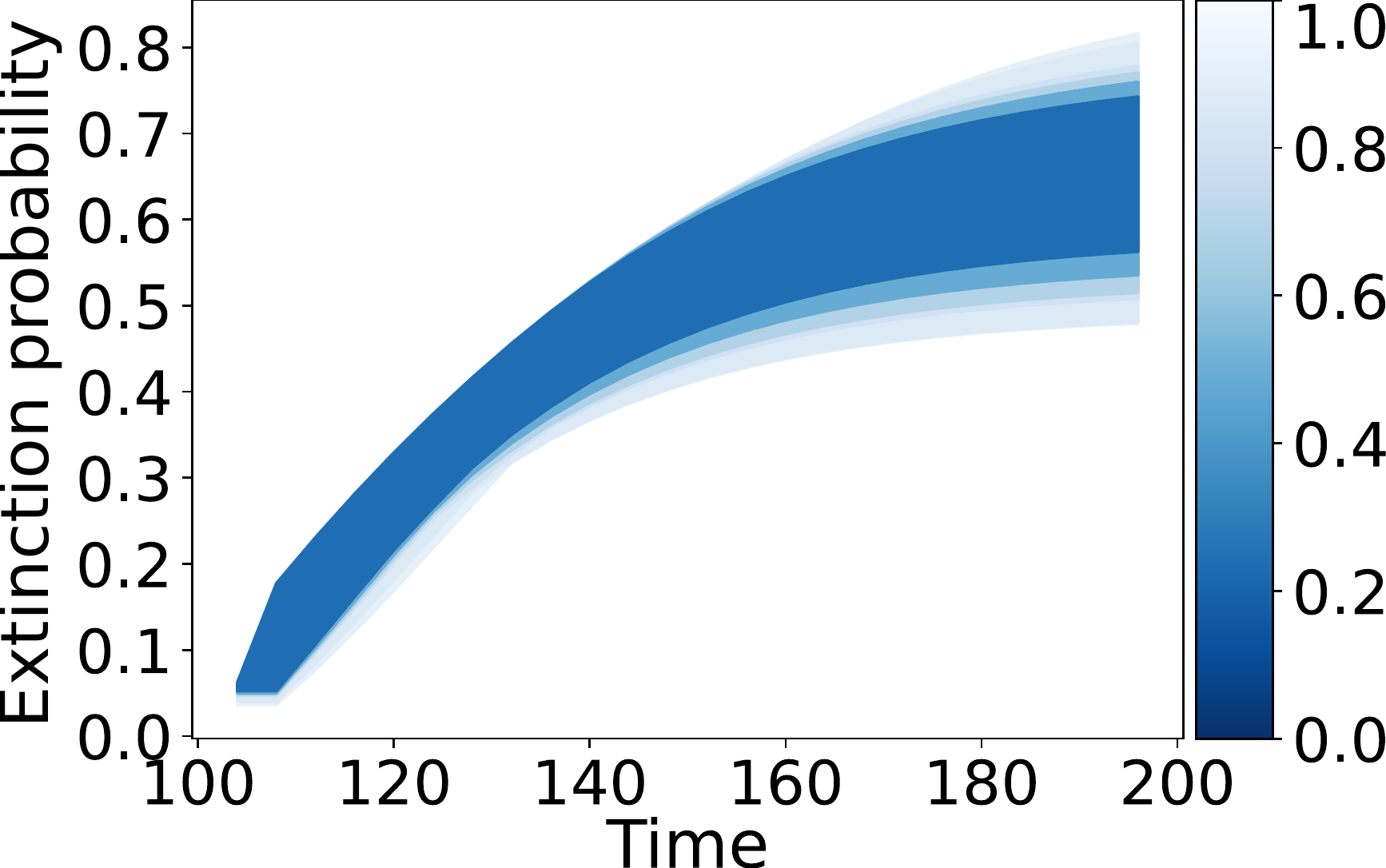}
  \captionof{figure}{Prediction regions for the SIR~(60) benchmark with $n=400$.}
  \label{fig:SIR_regions}   
\end{minipage}
\hfill
\begin{minipage}{.48\textwidth}
  \centering
  \includegraphics[width=.85\linewidth]{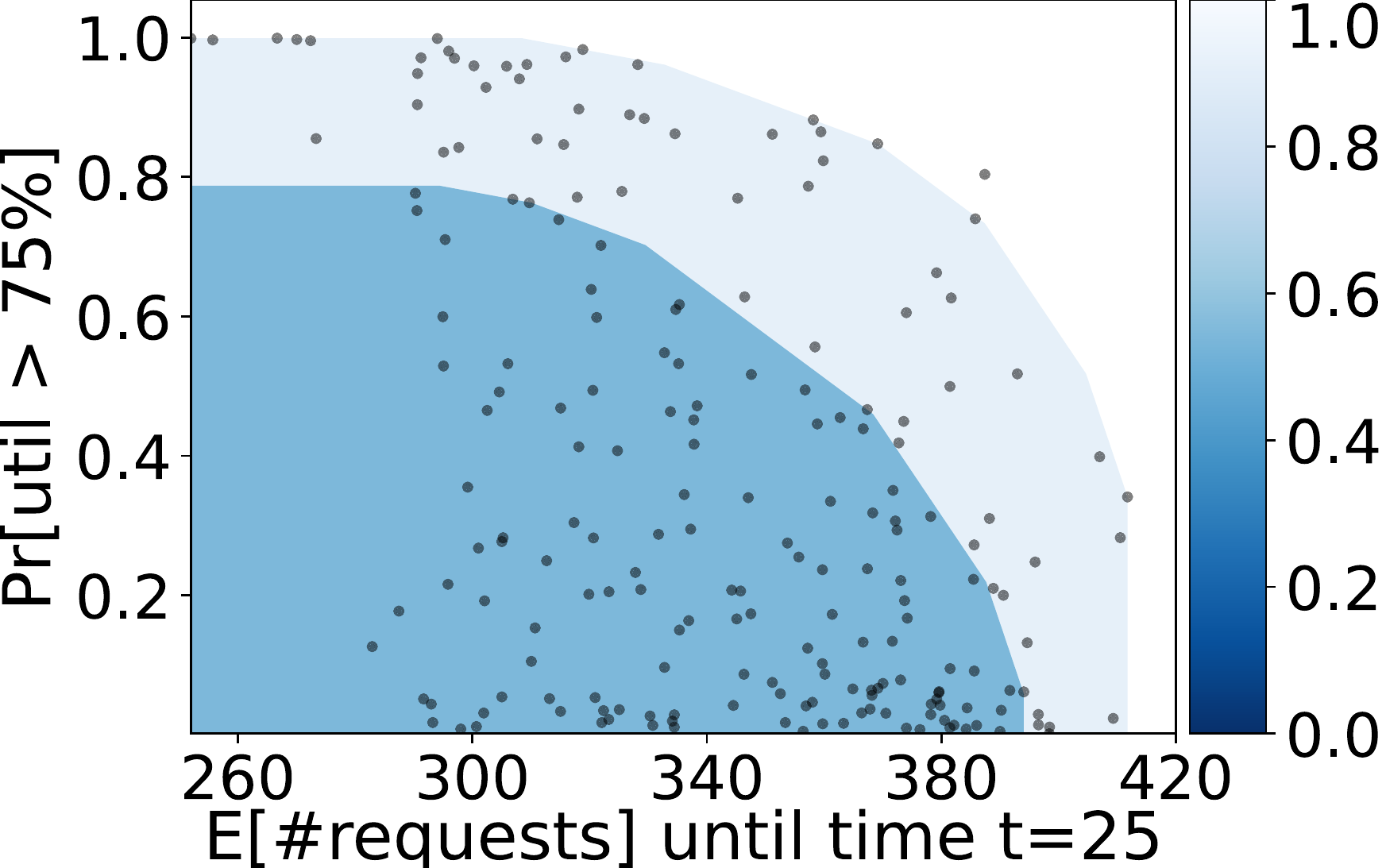}
  \captionof{figure}{Pareto front for the buffer benchmark with $n=200$ samples.}
  \label{fig:Buffer_pareto}
\end{minipage}
\end{figure}

\subsection*{Q1. Applicability}
An excerpt of the benchmark statistics is shown in \cref{tab:model_information} 
\ifappendix
    (see \cref{tab:model_information_full} in~\cref{sec:BenchmarkDetails} for the full table).
\else
    (see \cite[App.~C]{Badings2022extended} for the full table).
\fi
For all but the smallest benchmarks, sampling and computing the solution vectors by model checking is more expensive than solving the scenario problems. 
In the following, we illustrate that 100 samples are sufficient to provide qualitatively good prediction regions and associated lower bounds.

\smallpar{Plotting prediction regions.}
\cref{fig:SIR_regions} presents prediction regions on the extinction probability of the disease in the SIR model and is analogous to the tubes in \cref{fig:intro1c_confidence} 
\ifappendix
    (see \cref{fig:appendix_CTMC,fig:appendix_FT} in \cref{app:detailed_results} for plots for various other benchmarks).
\else
    (see \cite[App.~C.1]{Badings2022extended} for plots for various other benchmarks).
\fi
These regions are obtained by applying our algorithm with varying values for the cost of relaxation $\rho$.
For a confidence level of $\beta = 99\%$, the widest (smallest) tube in \cref{fig:SIR_regions} corresponds to a lower bound probability of $\mu = 91.1\%$ ($\mu=23.9\%$).
Thus, we conclude that, with a confidence of at least $99\%$, the curve created by the \gls{CTMC} for any sampled parameter value will lie within the outermost region in \cref{fig:SIR_regions} with a probability of at least $91.1\%$.
We highlight that our approach supports more general prediction regions. We show $n=200$ solution vectors for the buffer benchmark with two measures in \cref{fig:Buffer_pareto} and produce regions that approach the Pareto front.
For a confidence level of $\beta = 99\%$, the outer prediction region is associated with a lower bound probability of $\mu = 91.1\%$, while the inner region has a lower value of $\mu = 66.2\%$. 
We present more plots in 
\ifappendix
    \cref{app:detailed_results}.
\else
    \cite[App.~C.1]{Badings2022extended}.
\fi

\begin{table*}[t!]

\centering
\caption{Lower bounds $\bar{\mu}$ and standard deviation (SD), vs. the observed number of $1\,000$ additional solutions that indeed lie within the obtained regions.}

{\setlength{\tabcolsep}{3pt}

\begin{subtable}[h]{0.48\textwidth}
\caption{Kanban (3).}

\resizebox{\linewidth}{!}{%
\begin{tabular}{rrrrrl}
	\hline
 \rule{0pt}{1.5ex} 
  & \multicolumn{2}{c}{$\beta=0.9$}
  & \multicolumn{2}{c}{$\beta=0.999$}
  & Frequentist
  \\
  \cmidrule(lr){2-3}
  \cmidrule(lr){4-5}
  \cmidrule(lr){6-6}
  $n$
  & $\bar{\mu}$ & SD
  & $\bar{\mu}$ & SD
  & Observed
  \\
  \hline
100 & 0.862 & 0.000 & 0.798 & 0.000 & 959 $\pm$ 22.7 \\
200 & 0.930 & 0.000 & 0.895 & 0.000 & 967 $\pm$ 17.4 \\
400 & 0.965 & 0.001 & 0.947 & 0.001 & 984 $\pm$ 8.6 \\
800 & 0.982 & 0.000 & 0.973 & 0.000 & 994 $\pm$ 3.2 \\
\hline
\end{tabular}
}

\end{subtable}
\hfill
\begin{subtable}[h]{0.48\textwidth}
\caption{Railway crossing (1,1,hc).}

\resizebox{\linewidth}{!}{%
\begin{tabular}{rrrrrl}
	\hline
 \rule{0pt}{1.5ex} 
  & \multicolumn{2}{c}{$\beta=0.9$}
  & \multicolumn{2}{c}{$\beta=0.999$}
  & Frequentist
  \\
  \cmidrule(lr){2-3}
  \cmidrule(lr){4-5}
  \cmidrule(lr){6-6}
  $n$
  & $\bar{\mu}$ & SD
  & $\bar{\mu}$ & SD
  & Observed
  \\
  \hline
100 & 0.895 & 0.018 & 0.835 & 0.020 & 954 $\pm$ 26.8 \\
200 & 0.945 & 0.007 & 0.912 & 0.008 & 980 $\pm$ 12.8 \\
400 & 0.975 & 0.004 & 0.958 & 0.005 & 990 $\pm$ 8.3 \\
800 & 0.986 & 0.002 & 0.977 & 0.003 & 995 $\pm$ 4.3 \\
\hline
\end{tabular}
}

\end{subtable}

}

\label{tab:bound_tightness}
\end{table*}

\smallpar{Tightness of the solution.}
In \cref{tab:bound_tightness} we investigate the tightness of our results. For the experiment, we set $\rho=1.1$ and solve $\mathfrak{L}_\mathcal{U}^\rho$ for different values of $n$, repeating every experiment $10$ times, resulting in the average bounds $\bar{\mu}$.
Then, we sample $1\,000$ solutions and count the \emph{observed} number of solutions 
contained in every prediction regions, resulting in an empirical approximation of the containment probability.
Recall that for $\rho > 1$, we obtain a prediction region that contains all solutions, so this observed count grows toward $n$.
The lower bounds grow toward the empirical count for an increased $n$, with the smallest difference (RC, $n=800$, $\beta=0.9$) being as small as $0.9\%$. 
Similar observations hold for other values of $\rho$.

\smallpar{Handling imprecise solutions.}
The approximate model checker is significantly faster (see \cref{tab:model_information} for SIR (140) and RC), at the cost of obtaining imprecise solution vectors.\footnote{We terminate at a relative gap between upper/lower bound of the solution below 1\%.}
For SIR (140), the sampling time is reduced from $49$ to $\SI{9}{\minute}$, while the scenario optimization time is slightly higher at $\SI{129}{\second}$.
This difference only grows larger with the size of the \gls{CTMC}.
For the larger instances of RC and HECS, computing exact solutions is infeasible at all (one HECS (2,2) sample alone takes \SI{15}{\minute}).
While the bounds on the containment probability under imprecise solutions may initially be poor (see \cref{fig:Refinement_a}, which results in $\mu = 2.1\%$), we can improve the results significantly using the refinement scheme proposed in \cref{subsec:Refinement}.
For example, \cref{fig:Refinement_c} shows the prediction region after refining $31$ of the $100$ solutions, which yields $\mu = 74.7\%$.
Thus, \emph{by iteratively refining only the imprecise solutions on the boundary of the resulting prediction regions, we significantly tighten the obtained bounds on the containment probability}.

\begin{figure}[b]
    \subfloat[No solutions refined.]{%
    \includegraphics[height=3cm]{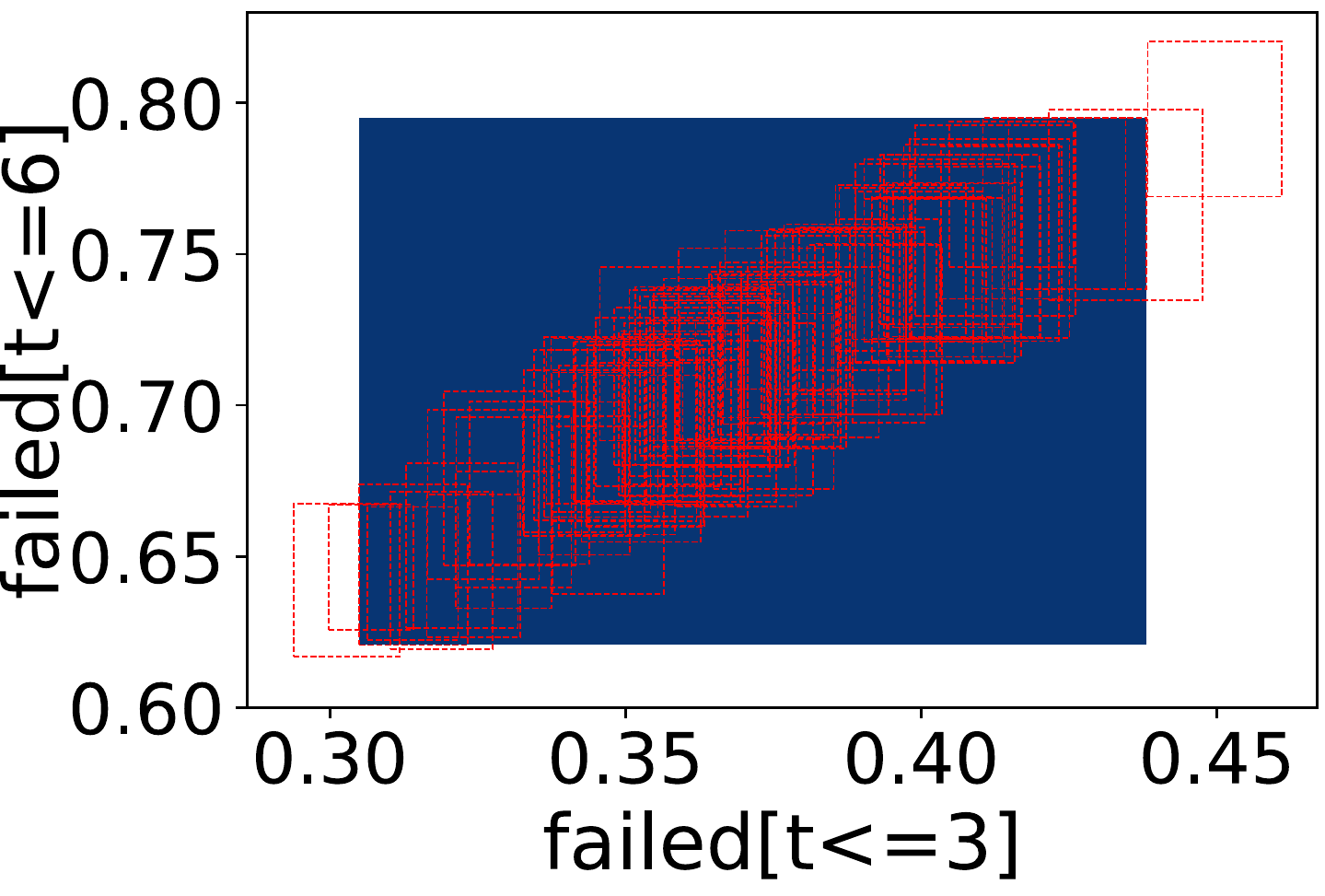}
    \label{fig:Refinement_a}}
    \subfloat[Intermediate step.]{%
    \includegraphics[height=3cm]{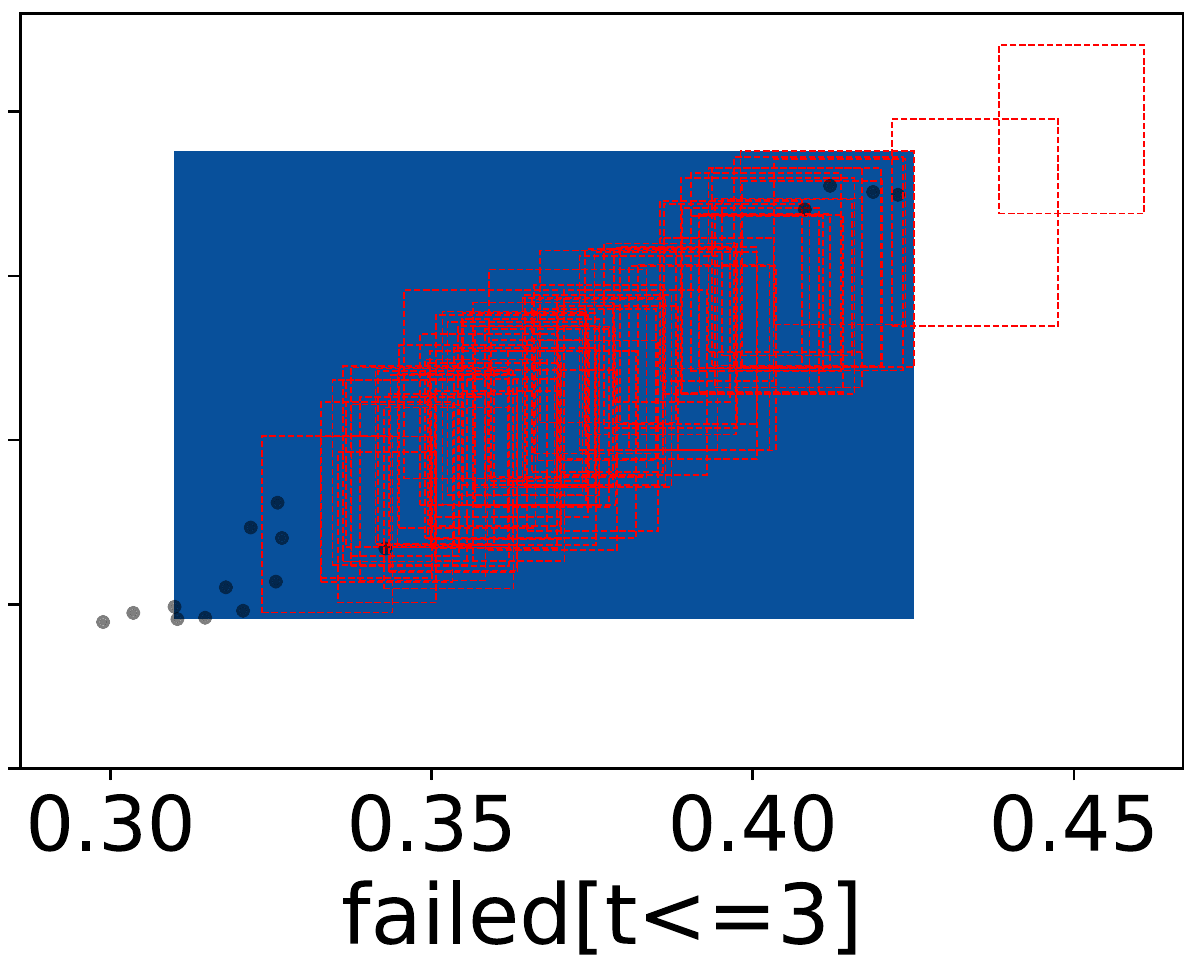}
    \label{fig:Refinement_b}}
    \subfloat[31 refined solutions.]{%
    \includegraphics[height=3cm]{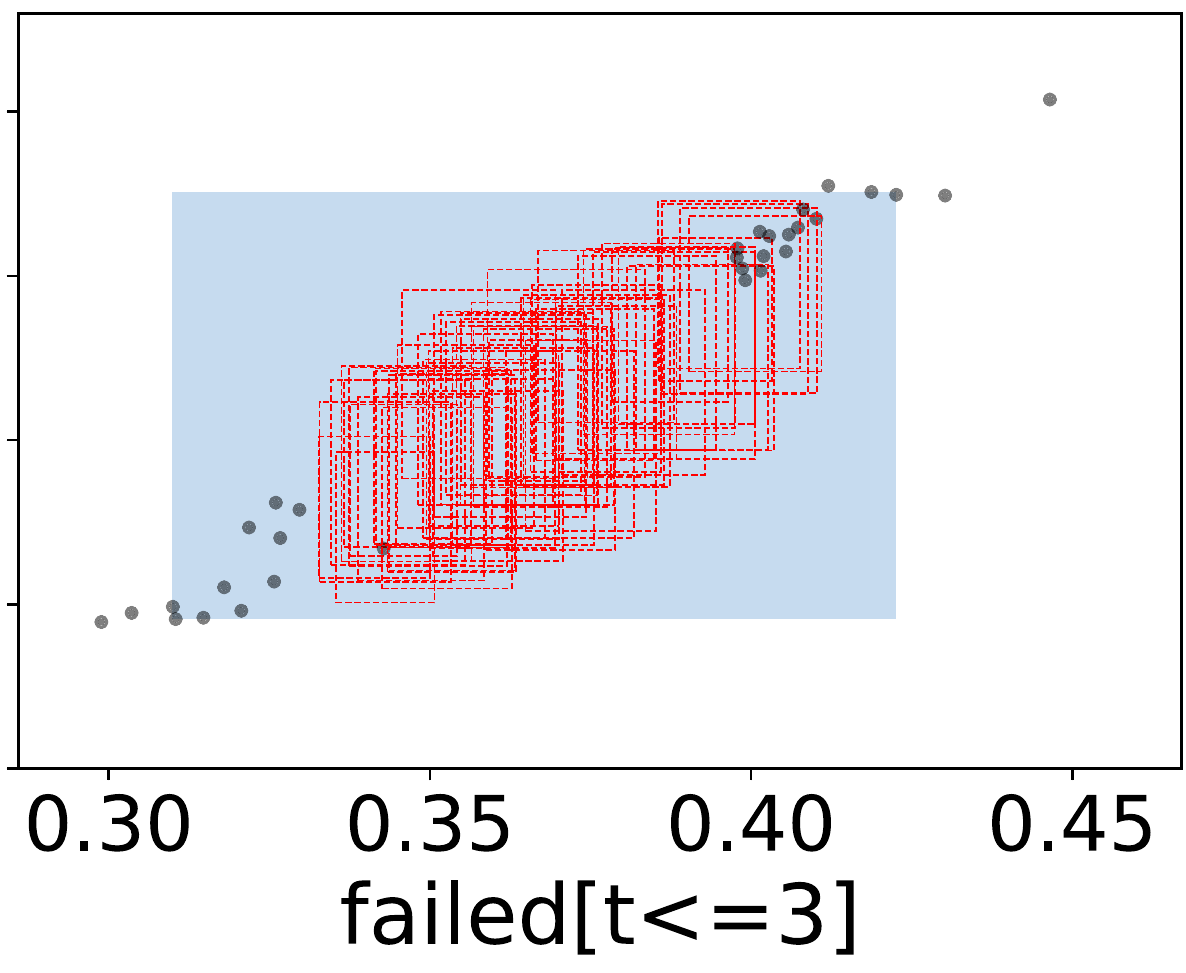}
    \label{fig:Refinement_c}}
    \caption{Refining imprecise solution vectors (red boxes) for RC (2,2), $n=100$.}
    \label{fig:Refinement}
\end{figure}

\begin{table*}[t!]

\centering
\caption{Run times in [s] for solving the scenario problems for SIR and RC with $\rho = 0.1$ (timeout (TO) of $1$ hour) for different sample sizes $n$ and measures $m$.}

{\setlength{\tabcolsep}{3pt}

\begin{subtable}[h]{0.48\textwidth}
\caption{SIR (population 20).}

\scalebox{0.9}{
\begin{tabular}{lrrrrr}
    \hline
    $n$ / $m$ & 50 & 100 & 200 & 400 & 800
    \\
    \hline
    100 &   0.97	    & 1.59	    & 3.36	    & 9.17	    & 25.41 \\
    200 &   3.69	    & 7.30	    & 22.91	    & 59.45	    & 131.78 \\
    400 &   29.43	    & 76.13	    & 153.03	& 310.67	& 640.70 \\
    800 &   261.97    & 491.73	& 955.77	& 1924.15   & TO \\
    \hline
\end{tabular}
}

\end{subtable}
\hfill
\centering
\begin{subtable}[h]{0.48\textwidth}
\caption{Railway crossing (1,1,hc).}

\scalebox{0.9}{
\begin{tabular}{lrrrr}
    \hline
    $n$ / $m$ & 50 & 100 & 200 & 400 
    \\
    \hline
    100 &   1.84	    & 3.40	    & 8.18      & 24.14 \\
    200 &   6.35	    & 14.56	    & 45.09     & 113.09 \\
    400 &   34.74       & 96.68	    & 203.77    & 427.80 \\
    800 &   292.32	    & 579.09	& 1215.67   & 2553.98 \\
    \hline
\end{tabular}
}

\end{subtable}

}

\label{tab:scalability_scenario}
\end{table*}

\subsection*{Q2. Scalability}
In \cref{tab:scalability_scenario}, we report the run times for steps (3)-(5) of our algorithm shown in \cref{fig:Approach} (i.e., for solving the scenario problems, but not for computing the solution vectors in \storm).
Here, we solve problem $\mathfrak{L}_\mathcal{U}^\rho$ for $\rho = 0.1$, with different numbers of samples and measures.
Our approach scales well to realistic numbers of samples (up to $800$) and measures (up to $400$).
The computational complexity of the scenario problems is largely \emph{independent of the size of the \gls{CTMC}}, and hence, similar run times are observed across the benchmarks (cf. \cref{tab:model_information}).

\subsection*{Q3. Comparison to baselines}
\label{subsec:eval_baselines}
We compare against two baselines:
(1)~Scenario optimization to analyze each measure independently, yielding a separate probabilistic guarantee on each measure.
(2)~A frequentist (Monte Carlo) baseline, which samples a large number of parameter values and counts the number of associated solutions within a region. 

\smallpar{Analyzing measures independently.}
To show that analyzing a full \emph{set of measures} at once, e.g., the complete probability curve, is essential, we compare our method to the baseline that analyzes \emph{each measure independently} and combines the obtained bounds on each measure afterward.
We consider the PCS benchmark with precise samples and solve $\mathfrak{L}_\mathcal{U}^\rho$ for $\rho=2$ 
\ifappendix
    (see \cref{tab:naive_scenario_baseline} in \cref{sec:BenchmarkDetails} for details).
\else
    (see \cite[Tab.~5]{Badings2022extended} for details).
\fi
For $n=100$ samples and $\beta=99\%$, our approach returns a lower bound probability of $\mu=84.8\%$.
By contrast, the na\"ive baseline yields a lower bound of only $4.5\%$, and similar results are observed for different values of $n$ 
\ifappendix
    (cf.~\cref{tab:naive_scenario_baseline}).
\else
    (cf.~\cite[Tab.~5~in~App.~C]{Badings2022extended}).
\fi
There are two reasons for this large difference.
First, the baseline applies \cref{thm:lower_bound_complexity} once for each of the $25$ measures, so it must use a more conservative confidence level of $\tilde{\beta} = 1 - \frac{1 - \beta}{25} = 0.9996$.
Second, the baseline takes the conjunction over the $25$ independent lower bounds, which drastically reduces the obtained bound.

\smallpar{Frequentist baseline.}
The comparison to the frequentist baseline on the Kanban and RC benchmarks yields the previously discussed results in \cref{tab:bound_tightness}.
The results in \cref{tab:model_information,tab:scalability_scenario} show that \emph{the time spent for sampling is (for most benchmarks) significantly higher than for scenario optimization.}
Thus, our scenario-based approach has a relatively low cost, while resulting in valuable guarantees which the baseline does not give.
To still obtain a high confidence in the result, a much larger sample size is needed for the frequentist baseline than for our approach. 
\section{Related Work}
\label{sec:Related}
Several verification approaches exist to handle uncertain Markov models.

For (discrete-time) \emph{interval} Markov chains (DTMCs) or \glspl{MDP}, a number of approaches verify against all probabilities within the intervals~\cite{DBLP:conf/lics/JonssonL91,DBLP:journals/ai/GivanLD00, Koushik-uncertainties,DBLP:journals/ijar/Skulj09,DBLP:conf/cav/PuggelliLSS13}.
Lumpability of interval CTMCs is considered in~\cite{DBLP:conf/qest/CardelliGLTTV21}.
In contrast to upCTMCs, interval Markov chains have no dependencies between transition uncertainties and no distributions are attached to the intervals. 

\emph{Parametric} Markov models generally define probabilities or rates via functions over the parameters.
The standard parameter synthesis problem for discrete-time models is to find all valuations of parameters that satisfies a specification. 
Techniques range from computing a solution function over the parameters, to directly solving the underlying optimization problems~\cite{daws2004symbolic,DBLP:journals/sttt/HahnHZ11,DBLP:journals/corr/abs-1903-07993,9640499}.
Parametric \glspl{CTMC} are investigated in \cite{DBLP:conf/rtss/HanKM08,DBLP:journals/acta/CeskaDPKB17}, but are generally restricted to a few parameters.
The work~\cite{DBLP:journals/jss/CalinescuCGKP18} aims to find a robust parameter valuation in \glspl{pCTMC}.

For all approaches listed so far, the results may be rather conservative, as no prior information on the uncertainties (the intervals) is used.
That is, the uncertainty is not quantified and all probabilities or rates are treated equally as likely.
In our approach, we do not compute solution functions, as the underlying methods are computationally expensive and usually restricted to a few parameters. 

Quantified uncertainty is studied in~\cite{DBLP:journals/sosym/MeedeniyaMAG14}.
Similarly to our work, the approach draws parameter values from a probability distribution over the model parameters and analyzes the instantiated model via model checking. 
However, \cite{DBLP:journals/sosym/MeedeniyaMAG14} 
studies DTMCs and performs a frequentist (Monte Carlo)
approach, \cf \cref{sec:Experiments}, to compute estimates for a single measure, without prediction regions.
Moreover, our approach requires significantly fewer samples, \cf the comparison in~\cref{subsec:eval_baselines}.

The work in \cite{DBLP:journals/iandc/BortolussiMS16,DBLP:conf/tacas/BortolussiS18} takes a sampling-driven Bayesian approach for \glspl{pCTMC}.
In particular, they take a prior on the solution function over a single measure and update it based on samples (potentially obtained via statistical model checking). 
We assume no prior on the solution function, and, as mentioned before, do not compute the solution function due to the expensive underlying computations.

\emph{Statistical model checking} (SMC)~\cite{DBLP:journals/tomacs/AghaP18,DBLP:series/lncs/LegayLTYSG19} samples path in stochastic models to perform model checking.
This technique has been applied to numerous models~\cite{DBLP:journals/sttt/DavidLLMP15,DBLP:conf/cav/DavidLLMW11,DBLP:conf/isola/DArgenioHS18,DBLP:journals/ress/RaoGRKVS09}, including 
\glspl{CTMC}~\cite{DBLP:conf/cav/SenVA05,DBLP:journals/iandc/YounesS06}.
SMC analyzes a \emph{concrete} \gls{CTMC} by sampling from the known transition rates, whereas for \gls{upCTMC} these rates are parametric.

Finally, scenario optimization~\cite{DBLP:journals/siamjo/CampiG08,Campi2021scenarioMaking} is widely used in control theory~\cite{DBLP:journals/tac/CalafioreC06} and recently in 
machine learning~\cite{DBLP:conf/cdc/CampiG20} and reliability engineering~\cite{rocchetta2021scenario}.
Within a verification context, closest to our work is~\cite{Badings2021ScenarioMDPs}, which considers the verification of single measures for uncertain \glspl{MDP}.
\cite{Badings2021ScenarioMDPs} relies on the so-called sampling-and-discarding approach~\cite{DBLP:journals/jota/CampiG11}, while we use the risk-and-complexity perspective~\cite{garatti2019risk}, yielding better results for problems with many decision variables like we have.
\section{Conclusion}
\label{sec:Conclusion}
This paper presents a novel approach to the analysis of parametric Markov models with respect to a set of performance characteristics. 
In particular, we provide a method that yields statistical guarantees on the typical performance characteristics from a finite set of samples of those parameters. 
Our experiments show that high-confidence results can be given based on a few hundred of samples. 
Future work includes supporting models with nondeterminism, exploiting aspects of parametric models such as monotonicity, and integrating methods to infer the distributions on the parameter space from observations.

\bibliographystyle{splncs04}
\bibliography{references,refsScenarioApproach}

\ifappendix
    \newpage
    \appendix
    \section{Additional Examples}
\label{app:Examples}

In the following example, we discuss the results shown in \cref{fig:confidence_2D_relaxed} in more detail.

\begin{example}
    \label{example:ViolationRho}
    We consider the Kanban manufacturing system benchmark from~\cite{ciardo1996use_KanbanBenchmark} with a Gaussian distribution over the parameters.
    In \cref{fig:confidence_2D_relaxed}, we present $n=25$ solution vectors for two expected cost measures.
    We use these solutions in problem $\mathfrak{L}^\rho_\mathcal{U}$ in \cref{eq:ScenarioProblem}, and solve for $\rho = 2$, $0.4$, and $0.15$.
    We show the three resulting prediction regions in \cref{fig:confidence_2D_relaxed}.
    For $\rho=2$, the prediction region contains all vectors, while for a lower cost of relaxation $\rho$, more vectors are left outside.
    \qed
\end{example}
    \section{Proofs}

\subsection{Proof of \cref{thm:RiskAndComplexity}}
\label{app:proof:thm1}

The proof of \cref{thm:RiskAndComplexity} is based on \cite{Campi2021scenarioMaking}, which states that for a complexity $c^*_\rho$ and for $\eta(c^*_\rho)$ the smallest positive solution to \cref{eq:RiskAndComplexity_polynomial}, it holds that 
\begin{equation}
    \label{eq:RiskAndComplexity_proof1}
    \amsmathbb{P}^n \Big\{ 
        V\big(R^*_\rho) \leq 1-\eta(c^*_\rho)
    \Big\}
    \geq \beta,
\end{equation}
where $V(R^*_\rho)$ is the so-called \emph{violation probability}, which is defined as
\begin{equation}
    \label{eq:RiskAndComplexity_proof2}
    V(R^*_\rho) = 
    \Pr \{ 
        u \in \paramspace \,\colon\, 
            \sol(u) \notin R^*_\rho
    \}.
\end{equation}
Observe that $\satprob(R^*_\rho) + V(R^*_\rho) = 1$.
Thus, we rewrite \cref{eq:RiskAndComplexity_proof1} as
\begin{equation}
    \label{eq:RiskAndComplexity_proof3}
    \amsmathbb{P}^n \Big\{ 
        \satprob(R^*_\rho) \geq \eta(c^*_\rho)
    \Big\}
    \geq \beta.
\end{equation}
Note that $\eta(c)$ is monotonically decreasing in $c$~\cite{garatti2021risk}, so for any $d^*_\rho \geq c^*_\rho$, we have $\eta(d^*_\rho) \leq \eta(c^*_\rho)$.
Hence, \cref{eq:RiskAndComplexity_proof3} also implies \cref{eq:RiskAndComplexity_bound}, which concludes the proof.

\subsection{Proof of \cref{lemma:Implication}}
\label{app:proof:implication}

Recall from the proof of \cref{thm:RiskAndComplexity} that for $\eta(c^*_\rho)$ the solution to \cref{eq:RiskAndComplexity_polynomial}, where $c^*_\rho$ is the true complexity of problem $\mathfrak{L}^\rho_\mathcal{U}$, it holds that
\begin{equation}
    \label{eq:ImpreciseSolutions_implication_proof1}
    \amsmathbb{P}^n \Big\{ 
        \satprob\big(R^*_\rho\big) \geq \eta(c^*_\rho)
    \Big\}
    \geq \beta.
\end{equation}
Observe that for any two sets $R^*_\rho \subseteq R$, we have $\satprob\big(R\big) \geq \satprob\big(R^*_\rho\big)$.
Moreover, recall that $\eta(c)$ is monotonically decreasing in $c$ (as also observed visually from \cref{fig:exampleBounds}), and thus, the condition $c^*_\rho \leq d$ implies that $\eta(d) \leq \eta(c_\rho^*)$.
Hence, under the proposed conditions, we rewrite \cref{eq:ImpreciseSolutions_implication_proof1} as the right-hand side of \cref{eq:ImpreciseSolutions_implication}, which concludes the proof.

\subsection{Proof of \cref{thm:subset_relation}}
\label{app:proof:subsetrelation}

\smallsubsubsection{The one-dimensional case.}
Let us first consider the case for one dimension, i.e., one measure. 
Recall from \cref{example:IntuitionRho} that for precise solutions in 1D, the \emph{outermost} two samples (labeled A and F) are relaxed under the (unique) optimal solution if $\rho < 1$, samples B and E if $\rho < \frac{1}{2}$, etc.
Denote by $\sol^-(u)_r = \sol^{-,\varphi_r}_{\instctmc{u}}$ and $\sol^+(u)_r = \sol^{+,\varphi_r}_{\instctmc{u}}$ the $r$-th entries of the respective imprecise solution vectors, \ie, the (upper and lower bound) solutions for measure $\varphi_r \in \Phi$ and the \gls{CTMC} created by sample $u \in \paramspace$.
In formalizing the relationship between the value of $\rho$ and whether a sample is relaxed, we state the following definition: 
\begin{definition}
    \label{def:solution_counts}
    Let $r \leq |\Phi|$. 
    For any $(\sol(u_i)^-, \sol(u_i)^+)$, we define $\thresh^+_\geq(u_i)_r \in \{1,\ldots,n\}$ and $\thresh^-_\leq(u_i)_r \in \{1,\ldots,n\}$ as the number of samples whose upper bound is at least $\sol^+(u_i)$ (or at most $\sol^-(u_i)$), when projected to dimension $r$: 
    \begin{align}
    \thresh^+_{\geq}(u_i)_r &= \big\vert \{ u_j \in \paramspace \, \colon \, \sol^+(u_j)_r \geq \sol^+(u_i)_r \} \big\vert
    \nonumber
    \\
    \thresh^-_{\leq}(u_i)_r &= \big\vert \{ u_j \in \paramspace \, \colon \, \sol^-(u_j)_r \leq \sol^-(u_i)_r \} \big\vert.
    \nonumber
    \end{align} 
\end{definition}

\begin{figure}[t!]
\centering
\iftikzcompile
      \resizebox{.42\linewidth}{!}{%
\begin{tikzpicture}[yscale=1.6]

\def\Z{0cm}
\def\hZ{0.3cm}

\def\A{0.2cm}
\def\hA{0.6cm}

\def\B{0.4cm}
\def\hB{1.0cm}

\def\C{0.8cm}
\def\hC{1.4cm}

\def\D{1.2cm}
\def\hD{1.8cm}

\def\E{2.0cm}
\def\hE{2.3cm}

\fill[MidnightBlue!20] (-0.3cm,\Z) rectangle (0cm,\hZ);
\fill[MidnightBlue!20] (0cm,\A) rectangle (0.3cm,\hA);
\fill[MidnightBlue!20] (-0.3cm,\B) rectangle (0cm,\hB);
\fill[MidnightBlue!20] (0cm,\C) rectangle (0.3cm,\hC);
\fill[MidnightBlue!20] (-0.3cm,\D) rectangle (0cm,\hD);
\fill[MidnightBlue!20] (0cm,\E) rectangle (0.3cm,\hE);

\draw[latex-] (0.3cm,\hE) -- (0.6cm,\hE) node[right, align=left] 
{${\sol}^+(u_1)$ \\  $\thresh^+_\geq(u_1) = 1$};
\draw[latex-] (-0.3cm,\hD) -- (-0.6cm,\hD) node[left, align=right] 
{${\sol}^+(u_2)$ \\  $\thresh^+_\geq(u_2) = 2$};
\draw[latex-] (0.3cm,\hC) -- (0.6cm,\hC) node[right, align=left] 
{${\sol}^+(u_3)$ \\  $\thresh^+_\geq(u_3) = 3$};
\draw[latex-] (-0.3cm,\hB) -- (-0.6cm,\hB) node[left, align=right] 
{${\sol}^+(u_4)$ \\  $\thresh^+_\geq(u_4) = 4$};
\draw[latex-] (0.3cm,\hA) -- (0.6cm,\hA) node[right, align=left] 
{${\sol}^+(u_5)$ \\ $\thresh^+_\geq(u_5) = 5$};
\draw[latex-] (-0.3cm,\hZ) -- (-0.6cm,\hZ) node[left, align=right] 
{${\sol}^+(u_6)$ \\ $\thresh^+_\geq(u_6) = 6$};

\draw[dashed,color=red!80] (0cm,\hB) -- (1.6cm,\hB) node[right] {$\bar{x}^+_\rho$};

\draw[-latex] (0cm,0.2cm) -- (0cm,2.5cm) node[pos=.055, below] {$\vdots$};

\end{tikzpicture}
}
\fi
  \captionof{figure}{The upper bounds of five imprecise solutions and their values of $\sol^+_\geq(u).$}
  \label{fig:1D_example_imprecise_num}
\end{figure}
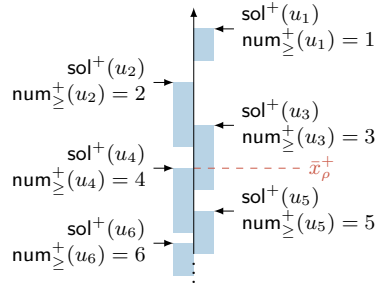

\noindent
\cref{fig:1D_example_imprecise_num} shows an extended version of \cref{fig:1D_example_imprecise} which also presents the values of $\thresh^+_{\geq}(u)$ for the upper bounds on the imprecise solutions.
The following lemma then characterizes the relaxed solutions:
\begin{lemma}
    \label{lemma:rho_relaxation}
    An imprecise solution $(\sol^-(u), \sol^+(u)), \, u \in \paramspace$ is not contained in the prediction region $[\munderbar{x}'_\rho, \bar{x}'_\rho]$ of $\mathfrak{G}^\rho_\mathcal{U}$, projected to dimension $p \leq \vert\Phi\vert$, if
    \begin{equation}
        \label{eq:rho_relaxation}
        \rho < \min\Big\{ 
        \thresh^+_\geq(u)_p, \,\, 
        \thresh^-_\leq(u)_p
    \Big\}^{-1}.
    \end{equation}
\end{lemma}
\noindent
For precise solutions, there is no distinction between $\sol^-(u_i)$ and $\sol^+(u_i)$, and \cref{def:solution_counts} yields the sets $\thresh_\geq(u_i)_r$ and $\thresh_\leq(u_i)_r$.
We can check whether a precise solution $\sol(u)$ is contained in the prediction region $[\munderbar{x}^*_\rho, \bar{x}^*_\rho]$ to $\mathfrak{L}^\rho_\mathcal{U}$ by replacing $\thresh^+ \to \thresh$ and $\thresh^- \to \thresh$ in \cref{eq:rho_relaxation}.

\begin{proof}[Proof of \cref{thm:subset_relation} (part I)]
    We proof the theorem by contradiction for the 1-dimensional case, and we generalize afterward.
    In a single dimension, $[\munderbar{x}^*_\rho, \bar{x}^*_\rho] \not\subseteq [\munderbar{x}'_\rho, \bar{x}'_\rho]$ requires that either $\munderbar{x}^*_\rho < \munderbar{x}'_\rho$ or $\bar{x}^*_\rho > \bar{x}'_\rho$.
    First consider the upper bounds $\bar{x}'_\rho$ and $\bar{x}^*_\rho$, which we can make explicit using \cref{lemma:rho_relaxation}:
    \begin{align}
        \bar{x}'_\rho &= \max\big\{ \sol^+(u), u \in \paramspace \,\colon\, \thresh^+_\geq(u) > \rho^{-1} \big\}
        \\
        \bar{x}^*_\rho &= \max\big\{ \sol(u), u \in \paramspace \,\colon\, \thresh_\geq(u) > \rho^{-1} \big\}.
    \end{align}
    For $\bar{x}^*_\rho > \bar{x}'_\rho$ to hold, the maximum $\thresh_\geq(u) > \rho^{-1}$ (i.e., the highest \emph{precise} solution for which there are more than $\rho^{-1}$ solutions at least as high) must exceed $\thresh^+_\geq(u) > \rho^{-1}$ (the highest \emph{imprecise upper bound} solution for which there are more than $\rho^{-1}$ imprecise upper bound solutions at least as high).
    This can only be true if the number of samples for which $\sol(u) > \bar{x}'$ is higher than the number for which $\sol^+(u) > \bar{x}'$.
    However, by construction, $\sol(u) \leq \sol^+$, so this is impossible, and thus, it holds that $\bar{x}^*_\rho \leq \bar{x}'_\rho$.
    While omitted for brevity, the proof that the lower bound $\munderbar{x}^*_\rho \geq \bar{x}'_\rho$ follows analogous to the upper bound.
\end{proof}

\smallsubsubsection{The multi-dimensional case.}

\begin{lemma}
    \label{remark:Independence_per_measure}
    Problem $\mathfrak{G}_\mathcal{U}^{\rho}$ can be decomposed into an independent problem for every dimension $1, \ldots, m$, with $m = \card{\Phi}$ the number of measures in $\Phi$.
\end{lemma}

\noindent
\cref{remark:Independence_per_measure} holds because the objective \cref{eq:ScenarioProblem_objective} is additive and all constraints for all measures \cref{eq:ScenarioProblem_constraint} are independent.
Thus, we can equivalently solve problem $\mathfrak{L}_\mathcal{U}^{\rho}$ for all $m$ measures separately. 

\begin{proof}[Proof of \cref{thm:subset_relation} (part II)]
    We now generalize the result to multiple dimensions.
    \cref{remark:Independence_per_measure} states that for rectangular prediction regions, problems $\mathfrak{L}^\rho_\mathcal{U}$ and $\mathfrak{G}^\rho_\mathcal{U}$ can be solved for each dimension separately.
    As such, we obtain an element-wise inequality $\munderbar{x}'_\rho \leq \munderbar{x}^*_\rho \leq \bar{x}^*_\rho \leq \bar{x}^*_\rho$, which also implies that $[\munderbar{x}^*_\rho, \bar{x}^*_\rho] \subseteq [\munderbar{x}'_\rho, \bar{x}'_\rho]$, so the claim in follows.
\end{proof}

\subsection{Proof of \Cref{thm:lower_bound_complexity}}
\label{app:proof:lower_bound_complexity}

Before providing the proof, we state the following useful lemma about surely noncritical samples:
\begin{lemma}
    \label{lemma:surely_noncritical}
    Any surely noncritical sample cannot be in the (smallest) critical set, defined in \cref{def:Complexity}.
\end{lemma}

\begin{proof}
    Recall from \cref{def:Complexity} that a sample may (potentially) be critical if it is either outside or on the boundary of the prediction region $[\munderbar{x}^*_\rho, \bar{x}^*_\rho]$.
    While the boundary of the prediction region $[\munderbar{x}^*_\rho, \bar{x}^*]$ is unknown, it cannot be smaller than the inner rectangle $\mathcal{I}$ defined in \cref{def:SurelyNoncritical}.
    By construction, any surely noncritical sample is a subset of this set $\mathcal{I}$.
    Hence, any surely noncritical sample cannot be in the (smallest) critical set, and the claim follows.
\end{proof}

\noindent
The proof of \cref{thm:lower_bound_complexity} now follows almost directly from \cref{lemma:surely_noncritical}.
The complexity is the cardinality of the smallest critical set, which cannot contain any surely noncritical sample, as stated by \cref{lemma:surely_noncritical}.
Hence, it follows that $n-\vert X \vert = \vert \sampleset \setminus \mathcal{X} \vert$, where $\mathcal{X}$ is the set of surely noncritical samples, which concludes the proof.
    \section{Detailed Benchmark Overview and Results}
\label{sec:BenchmarkDetails}

In this appendix, we illustrate on two specific benchmarks how we convert a \gls{pCTMC} into a \gls{upCTMC} by equipping its parameters with a probability distribution.
We omit details on the other benchmarks for brevity.
Thereafter, we present the full benchmark statistics in \cref{tab:model_information_full}, as well as more graphical outputs of our implementation in \cref{fig:appendix_CTMC,fig:appendix_FT}.

\smallsubsubsection{Epidemic modeling.}
We consider the classical \emph{SIR infection model}~\cite{andersson2012stochastic} of an infectious disease spreading through a population.
The factored state $s = (\bar{S},\bar{I},\bar{R})$ of this \gls{CTMC} counts the susceptible, infected, and recovered populations.
Infections and recoveries (we assume that immunization is permanent) occur based on the following parametric rules that depend on parameters $\lambda_i$ and $\lambda_r$: 
\begin{align}
    \textbf{Infection:\,\,} S + I \xrightarrow{\lambda_i \cdot \bar{S} \cdot \bar{I}} I + I, \quad\quad
    \textbf{Recovery:\,\,} I \xrightarrow{\lambda_r \cdot \bar{I}} R.
    \nonumber
\end{align}
In the classical SIR model, $\lambda_i$ and $\lambda_r$ are assumed to be known precisely, while we consider the parameters $\lambda_i = \mathcal{N}(0.05, 0.002)$ and $\lambda_r = \mathcal{N}(0.04, 0.002)$ to be normally distributed (recall that we only use samples from these distributions).
We define a set $\Phi = \{ \varphi_{100+100 i} \, \colon \, i = \frac{1}{m},\frac{2}{m},\ldots,\frac{m-1}{m},1 \}$ of $m$ measures, where $\varphi_{t}$ is the probability that the disease becomes extinct between time $100$ and $t$:
\begin{equation}
    \varphi_t = (\bar{I} > 0) \textbf{U}_{[100,t]} (\bar{I} = 0).
\end{equation}

\smallsubsubsection{Buffer system.}
We augment the \emph{producer-consumer buffering} system from~\cite{DBLP:journals/jss/CalinescuCGKP18}.
We equip the six parameters of this \gls{pCTMC} by uniform probability distributions with their domains specified below.
This \gls{pCTMC} models the transfer of requests from a producer (at a rate of $\lambda_g \in [32, 38]$) to consumers who consume them (at a rate of $\lambda_c \in [27, 33]$).
The requests are sent at a rate of $\lambda_t \in [27, 33]$, via either a slow or a fast buffer, with probabilities of $0.6$ and $0.4$, respectively.
While being faster, the fast buffer is less reliable than the slow buffer (it loses requests with a probability $\lambda_\text{loss} \in [0.025, 0.075]$), and has a smaller capacity.
Requests from the slow buffer are transferred to the fast buffer with a probability proportional to the occupancy.
The transmission rate of the slow buffer is $\lambda_{\text{slow}} \in [5, 15]$; the rate of the fast buffer is $\lambda_\delta \in [5, 15]$ higher.
We consider two measures: (1) the expected transferred requests until time $25$, and (2) the probability that the utilization of both buffers is above 75\% within the time $[20,25]$.

\subsection{Detailed results}
\label{app:detailed_results}

The complete overview of the statistics of all benchmarks is shown in \cref{tab:model_information_full}.
Running polling (15) for $n=200$ samples led to a timeout, due to the very high sampling times (sampling 100 solutions already takes over 3.5 hours).
The obtained prediction regions for eight benchmarks (under precise solutions) are presented in \cref{fig:appendix_CTMC,fig:appendix_FT}.
These plots demonstrate the wide applicability and effectiveness of our method on a large variety of benchmarks.

{
\setlength{\tabcolsep}{3pt}
\begin{table*}[t!]

\centering
\caption{Model sizes and run times for all benchmarks (the sampling time in \storm is reported per $100$ \glspl{CTMC}).}

\begin{threeparttable}
\scalebox{0.90}{
\begin{tabular}{lrrrrrrrr}
	\hline
 \rule{0pt}{1.5ex} & & \multicolumn{3}{c}{{Model size}}
  & \multicolumn{2}{c}{{\storm run time [s]}} & \multicolumn{2}{c}{{Scen.opt. time [s]}} \\
  \cmidrule(lr){3-5} \cmidrule(lr){6-7} \cmidrule(lr){8-9}
  benchmark
  & $\vert\Phi\vert$
  & \#pars
  & \#states
  & \#trans 
  & Init.
  & Sample ($\times100$)
  & $N=100$
  & $N=200$
  \\
  \hline
SIR (20)    & 26    & 2     & 216           & 396           & 0.02      & 2.93      & 2.44      & 10.18 \\
SIR (60)    & 26    & 2     & 1\,876        & 3\,636        & 0.06      & 121.65    & 16.91     & 52.57 \\
SIR (100)   & 26    & 2     & 5\,136        & 10\,076       & 0.15      & 829.47    & 19.35 & 62.74 \\
SIR (100)\tnote{a} & 26 & 2 & 5\,136        & 10\,076       & 0.15      & 191.30    & 27.45     & 137.76 \\
SIR (140)   & 26    & 2     & 9\,996        & 19\,716       & 0.29      & 2947.29   & 18.26     & 63.27  \\
SIR (140)\tnote{a} & 26 & 2 & 9\,996        & 19\,716       & 0.29      & 544.27    & 25.11     & 129.66  \\
Kanban (3)  & 4     & 13    & 58\,400       & 446\,400      & 4.42      & 46.95     & 2.28      & 6.69   \\
Kanban (5)  & 4     & 13    & 2\,546\,432   & 24\,460\,016  & 253.39    & 4363.63   & 2.03      & 5.94    \\
polling (3) & 2     & 2     & 36            & 84            & 0.02      & 0.08      & 2.35      & 7.31     \\
polling (9) & 2     & 2     & 6\,912        & 36\,864       & 0.64      & 22.92     & 2.13      & 6.66     \\
polling (15)& 2     & 2     & 737\,280      & 6\,144\,000   & 3\,908.13 & 9\,509.22 & 2.07      & --        \\
buffer      & 2     & 6     & 5\,632        & 21\,968       & 0.48      & 20.70     & 1.21      & 4.15      \\
tandem (15) & 2     & 5     & 496           & 1\,619        & 0.03      & 82.01     & 1.67      & 5.36      \\
tandem (31) & 2     & 5     & 2\,016        & 6\,819        & 0.11      & 862.41    & 5.19      & 24.30     \\
embed. (64) & 3     & 6     & 55\,868       & 235\,793      & 3.31      & 8.07      & 3.97      & 12.69     \\
embed. (256) & 3    & 6     & 218\,108      & 920\,657      & 18.04     & 33.40     & 3.99      & 13.03     \\
pcs         & 25    & 6     & 2\,501        & 14\,985       & 0.03      & 37.20     & 3.51      & 13.07     \\
rbc         & 40    & 6     & 2\,269        & 12\,930       & 0.01      & 1.40      & 5.27      & 16.88     \\
dcas        & 25    & 10    & 64            & 202           & 0.01      & 0.48      & 3.10      & 10.88     \\
rc (1,1)    & 25    & 21    & 8\,401        & 49\,446       & 27.20     & 74.90     & 5.75      & 20.34     \\
rc (1,1)\tnote{a}   & 25    & 21    & n/a\tnote{b}   & n/a\tnote{b}           & 0.02      & 2.35      & 29.23     & 150.61    \\
rc (2,2)\tnote{a}   & 25    & 29    & n/a\tnote{b}   & n/a\tnote{b}           & 0.03      & 27.77     & 24.86     & 132.63    \\
hecs (2,1)  & 25    & 5     & 118\,945      & 1\,018\,603   & 0.04      & 6.42      & 3.87      & 13.18     \\
hecs (2,1)\tnote{a} & 25    & 5     & n/a\tnote{b} & n/a\tnote{b} & 0.02      & 9.83      & 26.78     & 145.77    \\
hecs (2,2)\tnote{a} & 25    & 24    & n/a\tnote{b} & n/a\tnote{b} & 0.02      & 194.25    & 33.06     & 184.32    \\
\hline
\end{tabular}
}
\begin{tablenotes}
        \raggedright
        \item[a]Computed using approximate model checking up to a relative gap between upper\\bound $\sol^+(u)$ and lower bound $\sol^-(u)$ below $1\%$ for every sample $u \in \paramspace$.
        \item[b]Model size is unknown, as the approximation does not build the full state-space.
\end{tablenotes}
\end{threeparttable}

\label{tab:model_information_full}
\end{table*}
}

\smallsubsubsection{Analyzing measures independently.}
\label{app:results_independent_measures}
\begin{table*}[t!]

\centering
\caption{Obtained bounds (for the PCS fault tree) on the containment probability for our approach and the baseline that analyzes each measure independently.}

{\setlength{\tabcolsep}{4pt}

\resizebox{\linewidth}{!}{%
\begin{tabular}{lllllllll}
	\hline
 \rule{0pt}{1.5ex} 
  & \multicolumn{2}{c}{$n=100$}
  & \multicolumn{2}{c}{$n=200$}
  & \multicolumn{2}{c}{$n=400$}
  & \multicolumn{2}{c}{$n=800$}
  \\
  \cmidrule(lr){2-3}
  \cmidrule(lr){4-5}
  \cmidrule(lr){6-7}
  \cmidrule(lr){8-9}
  Method
  & $\beta=0.9$ & $\beta=0.999$
  & $\beta=0.9$ & $\beta=0.999$
  & $\beta=0.9$ & $\beta=0.999$
  & $\beta=0.9$ & $\beta=0.999$
  \\
  \hline
    Our approach    & 0.908 &	0.848 &	0.937 &	0.903 &	0.976 &	0.960 &	0.984 &	0.975 \\
    Baseline   & 0.045	&   0.010 &	0.212 &	0.103 &	0.461 &	0.322 &	0.679 &	0.567 \\
\hline
\end{tabular}
}

}
\label{tab:naive_scenario_baseline}
\end{table*}
\cref{tab:naive_scenario_baseline} presents the full comparison on the PCS benchmark between our approach and the baseline scenario approach that analyzes each measure independently.
In this table, we report the average lower bounds (over 10 iterations) on the containment probability, for different sample sizes $n = 100,\ldots,800$ and confidence levels $\beta$.
For the two main reasons for the significant difference in the tightness of the containment probability, we refer to \cref{sec:Experiments} in the main paper. 

\begin{figure}[t!]
\centering
    \begin{subfigure}[b]{0.45\linewidth}
        \centering
        \includegraphics[width=\linewidth]{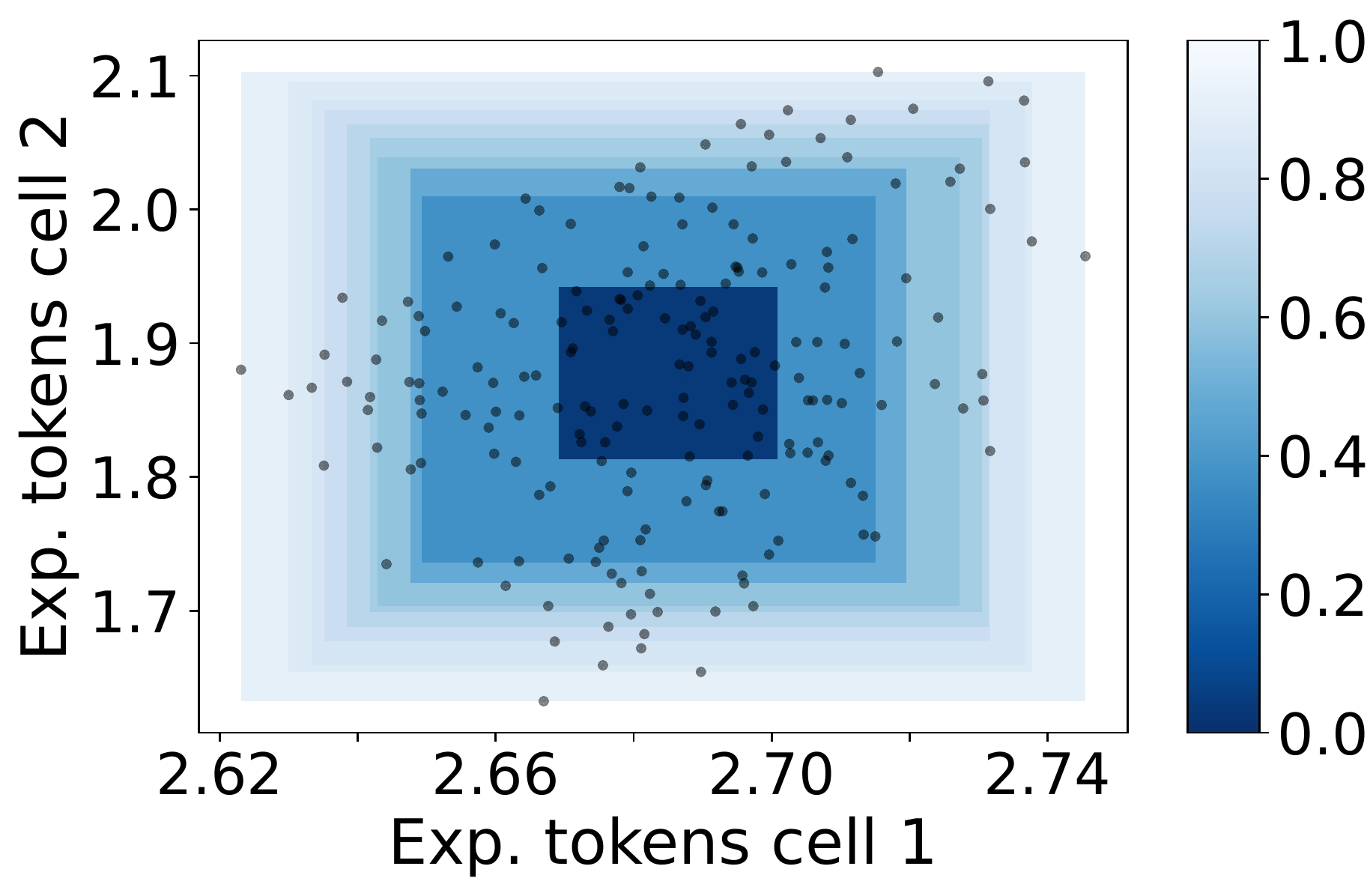}
        \captionof{figure}{Kanban (3), $n=200$.}
        \label{fig:appendix_CTMC_kanban}
    \end{subfigure}
    \hfill
    \begin{subfigure}[b]{0.45\linewidth}
        \centering
        \includegraphics[width=\linewidth]{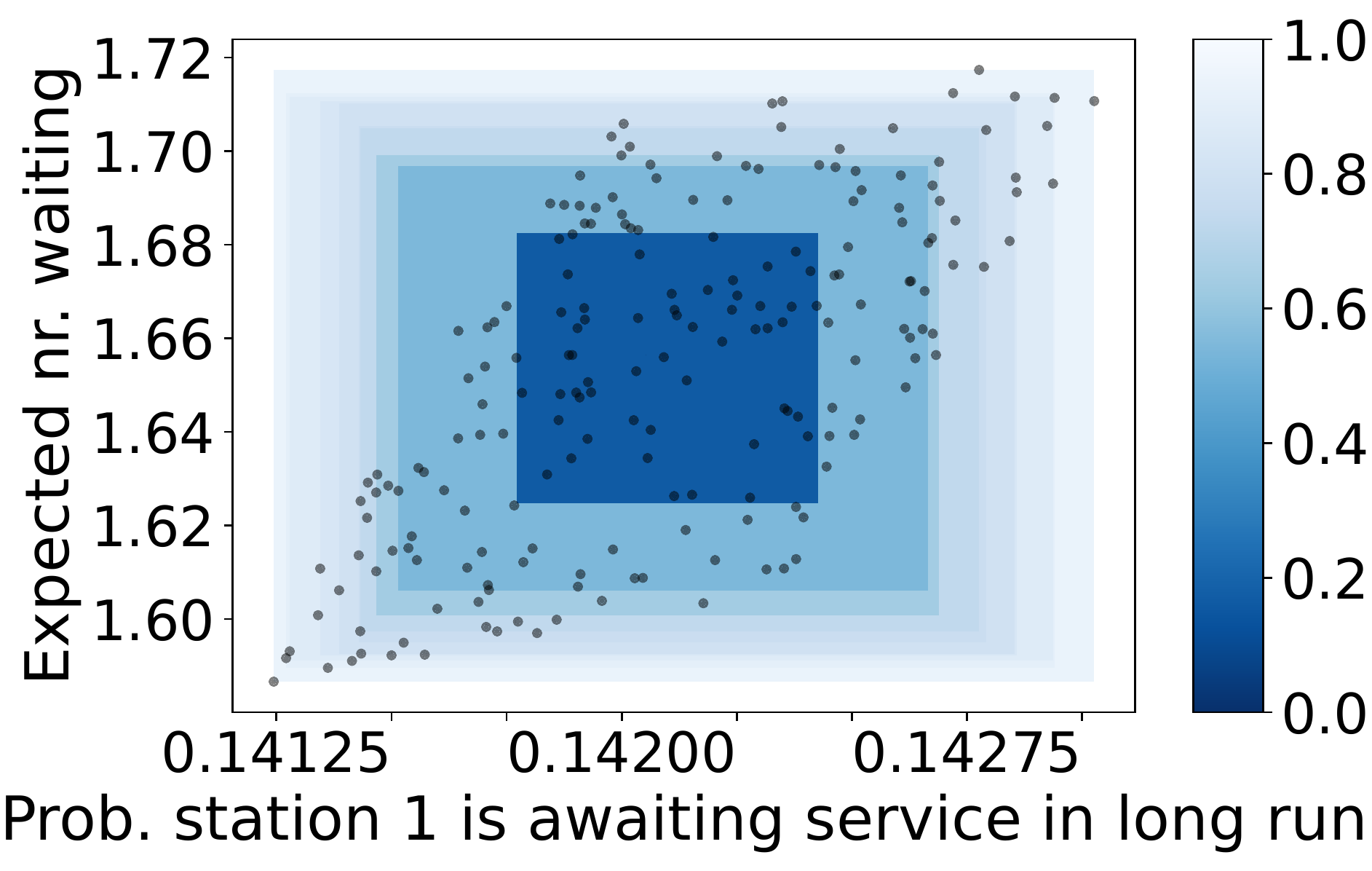}
        \captionof{figure}{Polling (9), $n=200$.}
        \label{fig:appendix_CTMC_polling}
    \end{subfigure}
    \hfill
    \begin{subfigure}[b]{0.45\linewidth}
        \centering
        \includegraphics[width=\linewidth]{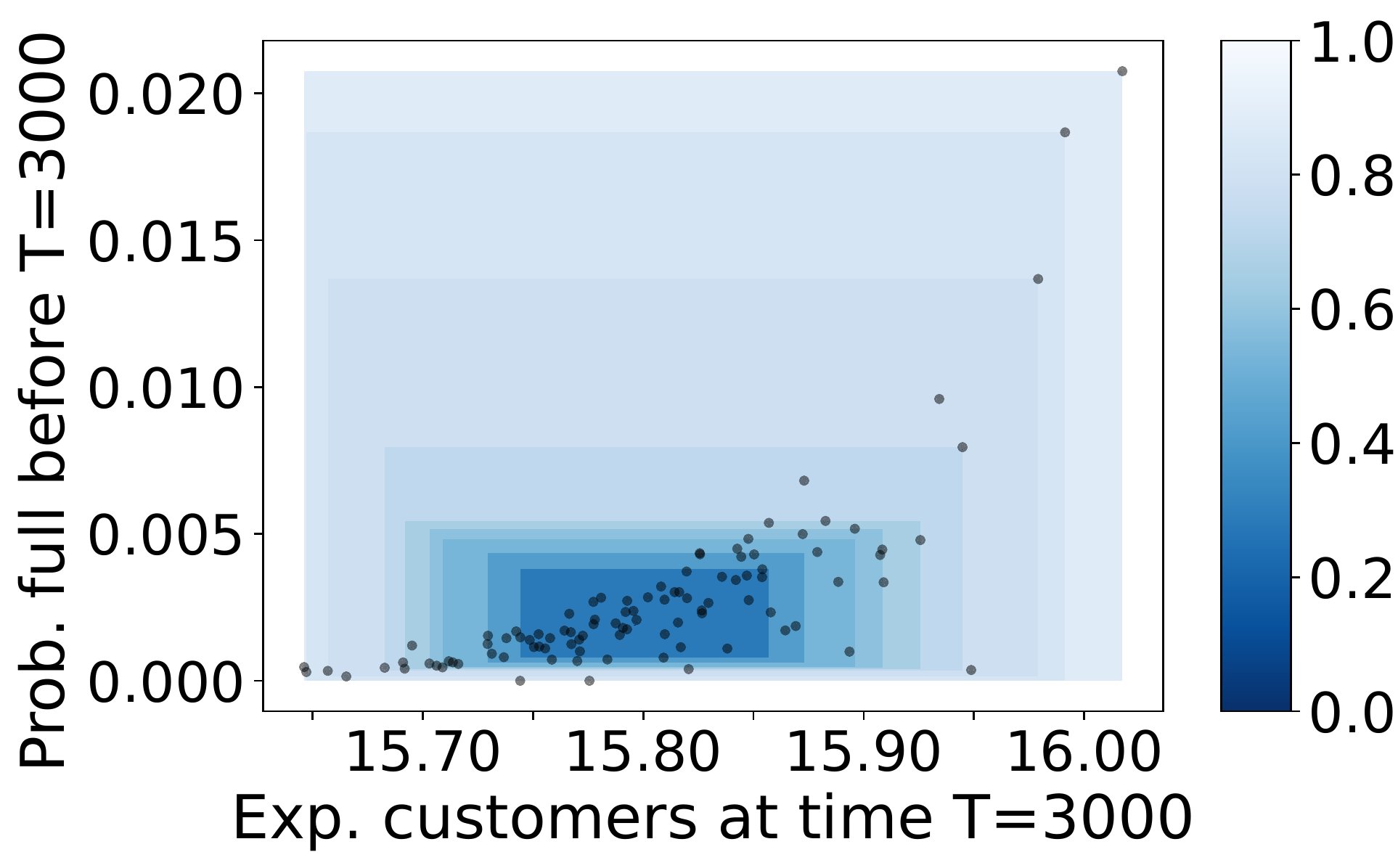}
        \captionof{figure}{Tandem (15), $n=100$.}
        \label{fig:appendix_CTMC_tandem}
    \end{subfigure}
    \hfill
    \begin{subfigure}[b]{0.45\linewidth}
        \centering
        \includegraphics[width=\linewidth]{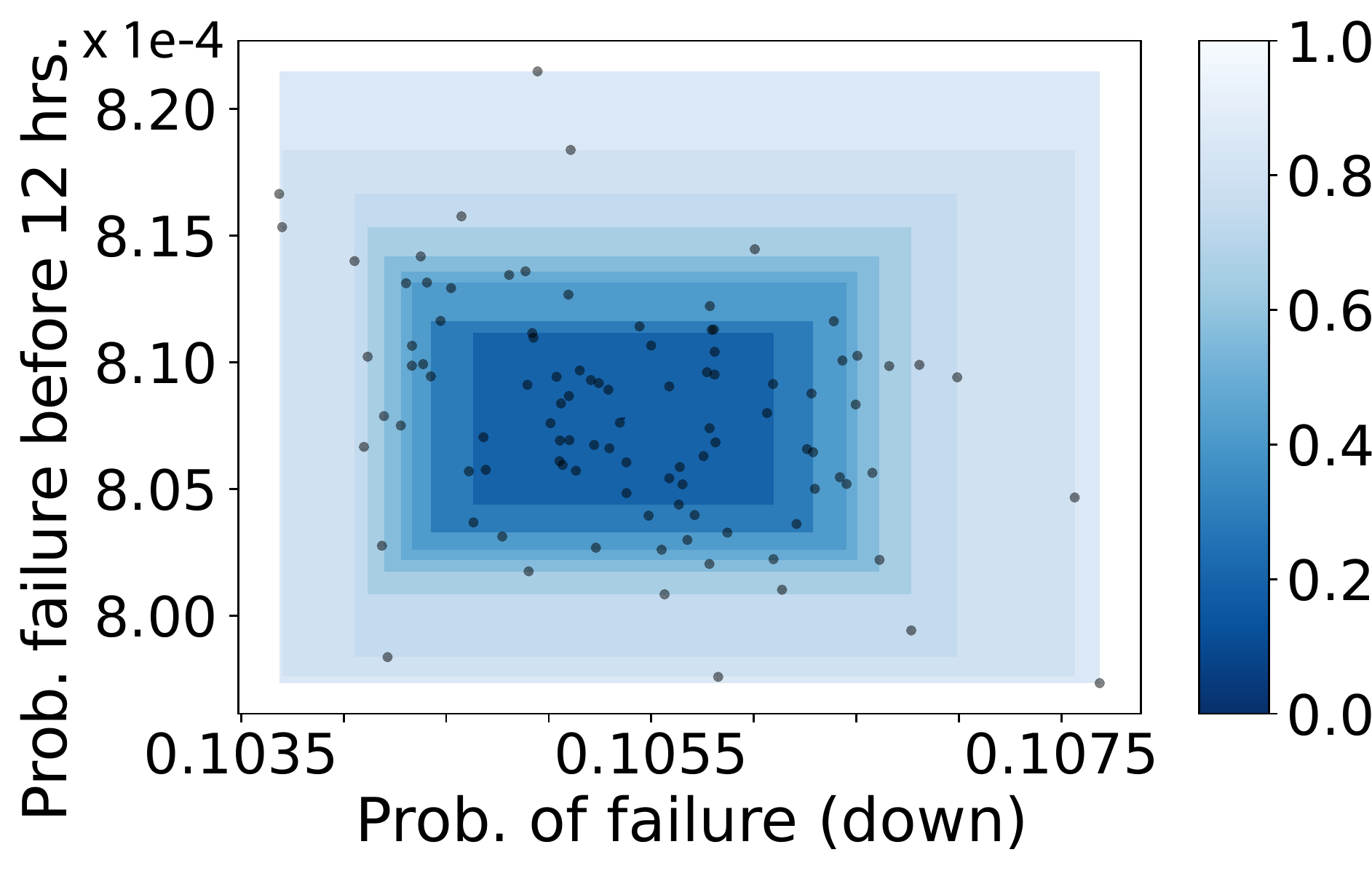}
        \captionof{figure}{Embedded (64), $n=100$.}
        \label{fig:appendix_CTMC_embedded}
    \end{subfigure}
\caption{Prediction regions for \gls{CTMC} benchmarks.} 
\label{fig:appendix_CTMC}
\end{figure}

\begin{figure}[t!]
\centering
    \begin{subfigure}[b]{0.45\linewidth}
        \centering
        \includegraphics[width=\linewidth]{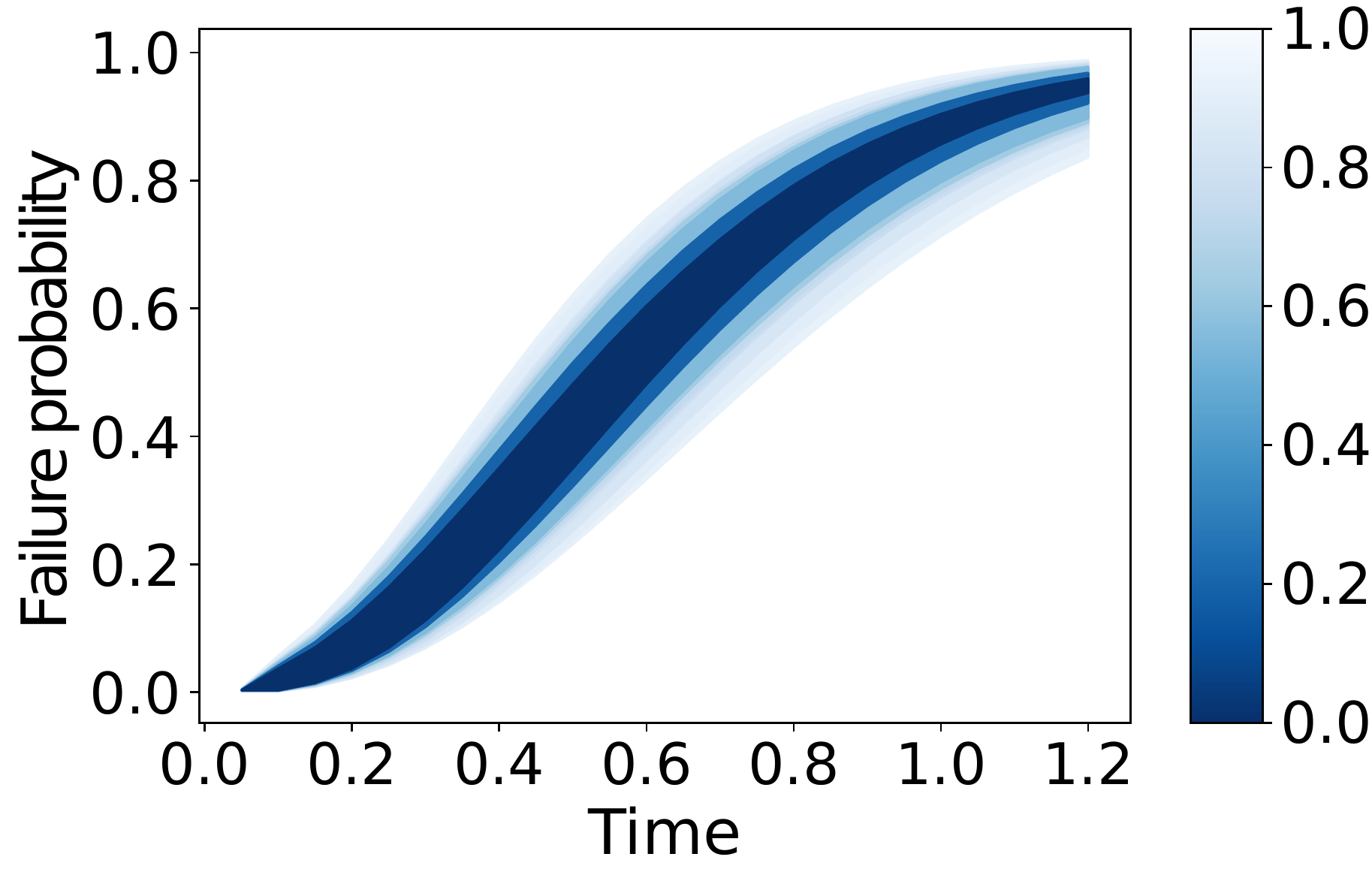}
        \captionof{figure}{PCS, $n=200$.}
        \label{fig:appendix_FT_cs}
    \end{subfigure}
    \hfill
    \begin{subfigure}[b]{0.45\linewidth}
        \centering
        \includegraphics[width=\linewidth]{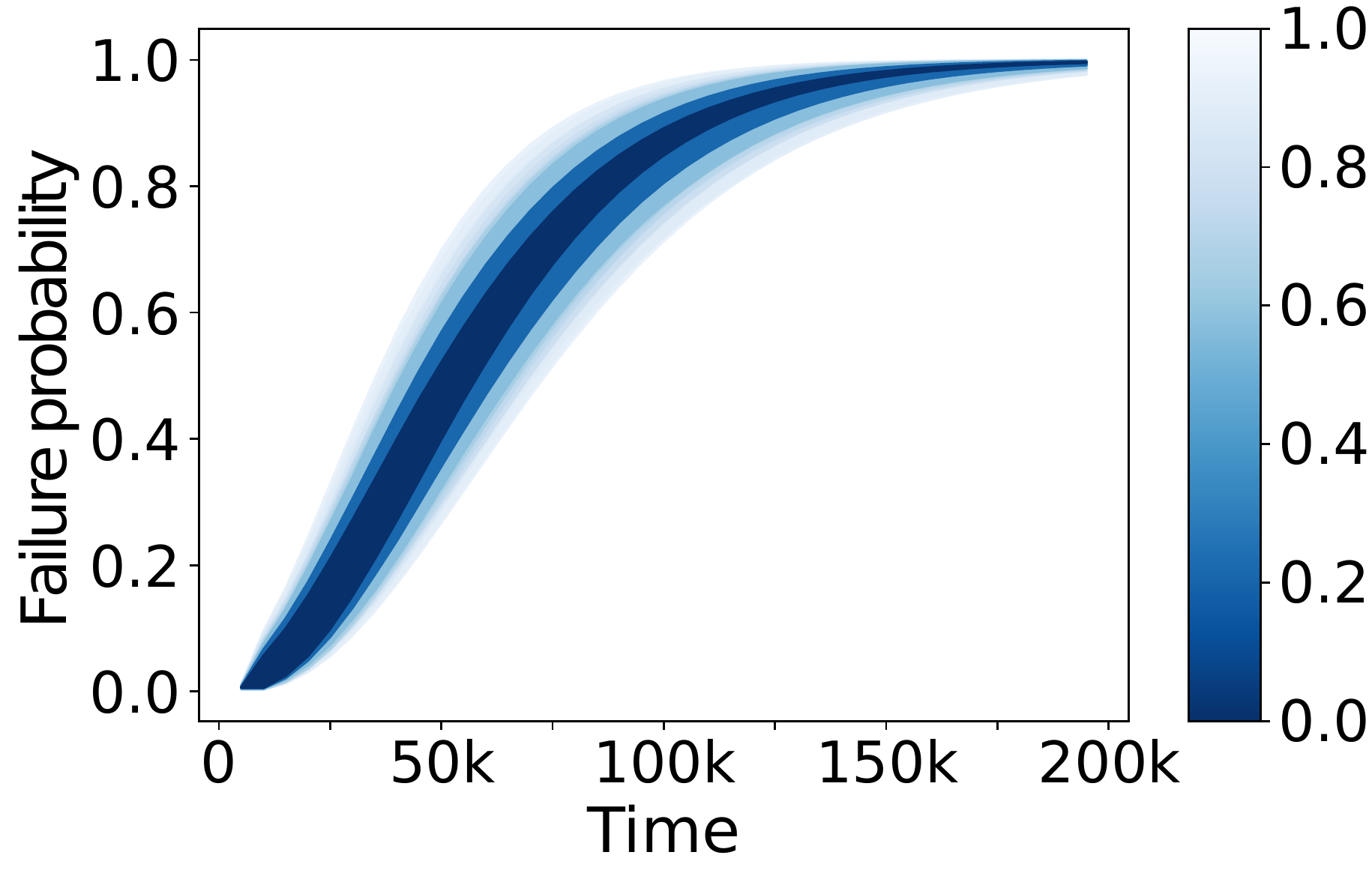}
        \captionof{figure}{RBC, $n=200$.}
        \label{fig:appendix_FT_rbc}
    \end{subfigure}
    \hfill
    \begin{subfigure}[b]{0.45\linewidth}
        \centering
        \includegraphics[width=\linewidth]{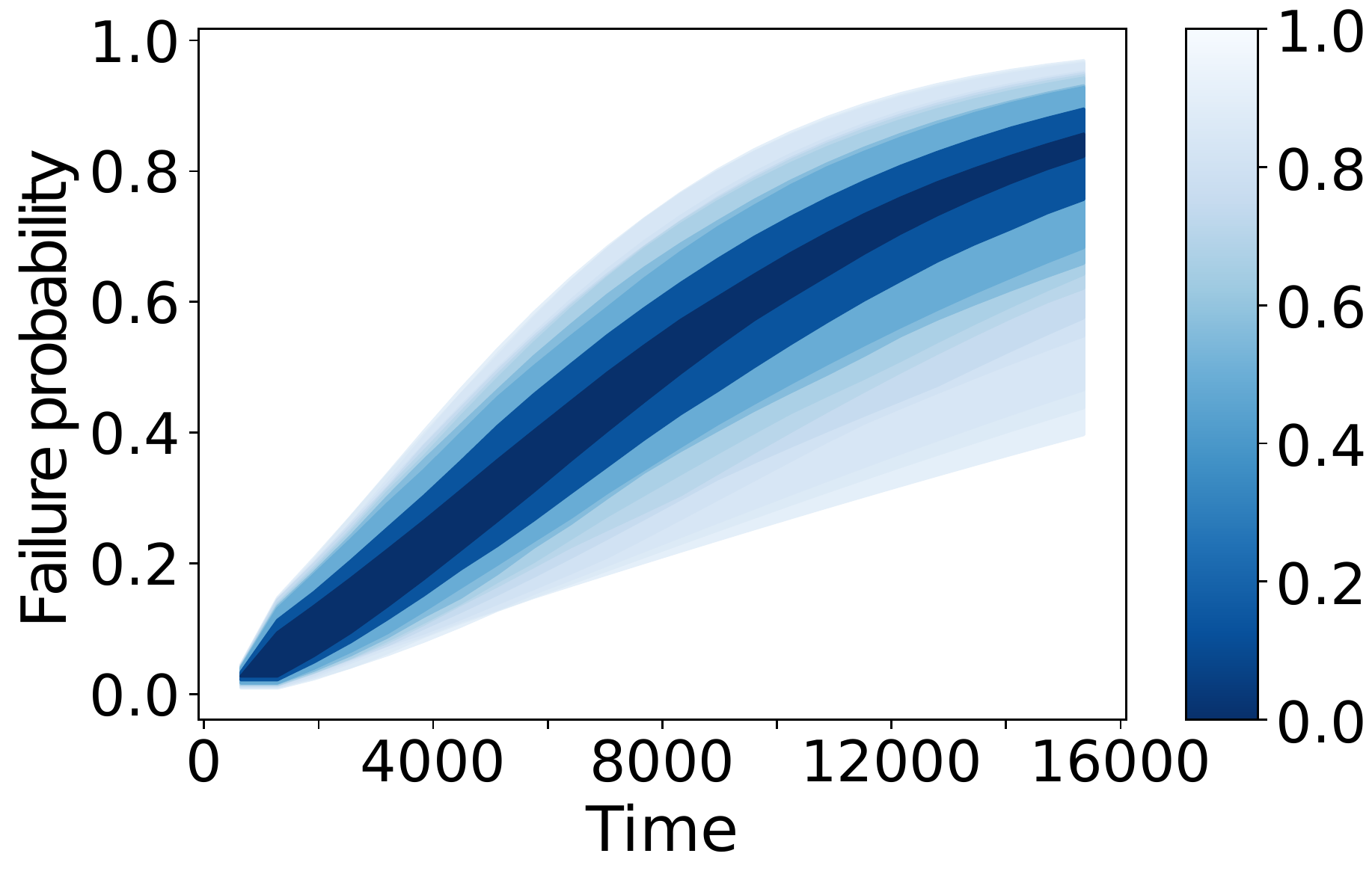}
        \captionof{figure}{DCAS, $n=200$.}
        \label{fig:appendix_FT_dcas}
    \end{subfigure}
    \hfill
    \begin{subfigure}[b]{0.45\linewidth}
        \centering
        \includegraphics[width=\linewidth]{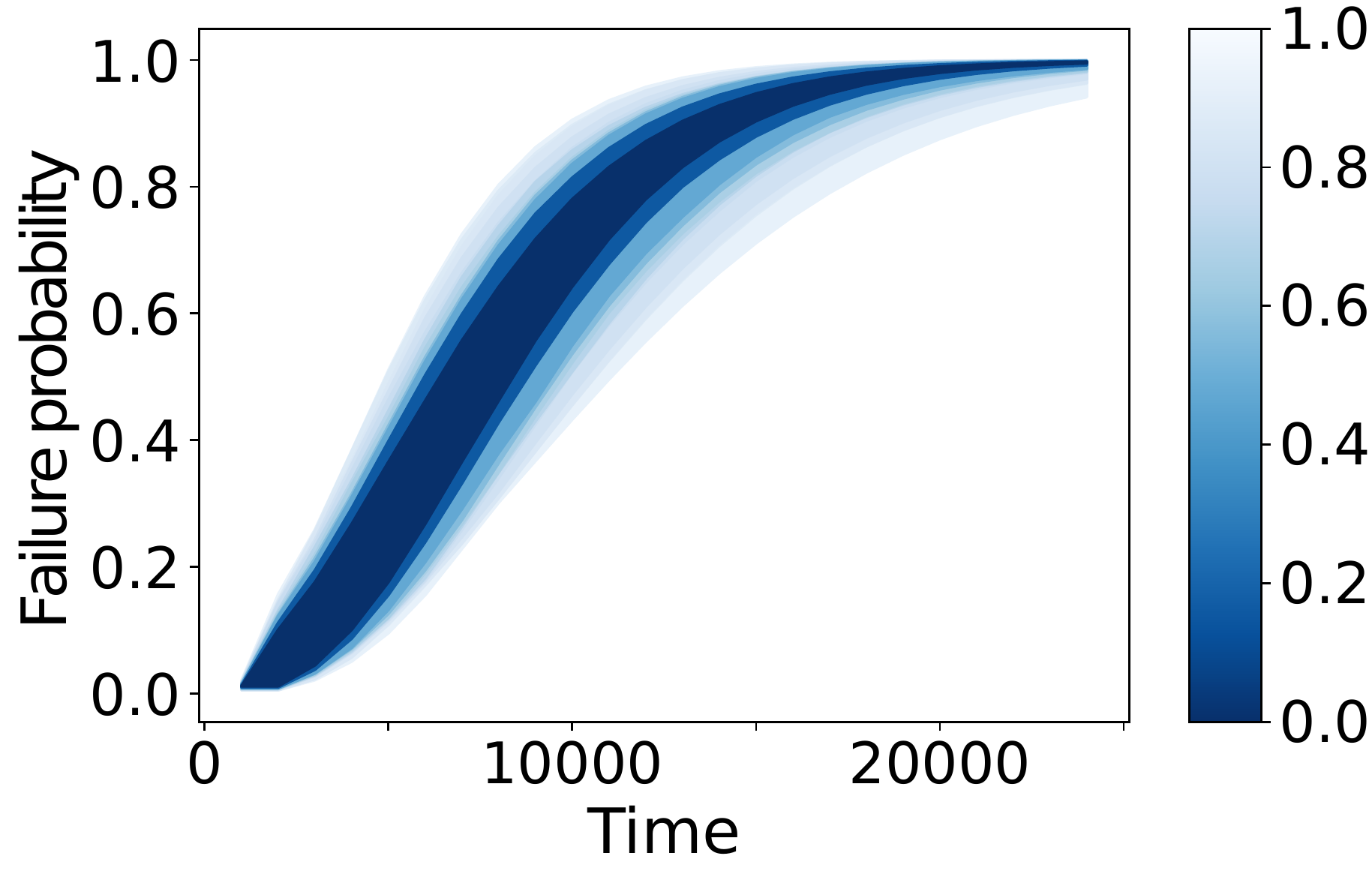}
        \captionof{figure}{HECS, $n=200$.}
        \label{fig:appendix_FT_hecs}
    \end{subfigure}
\caption{Prediction regions on probability curves for fault tree benchmarks.} 
\label{fig:appendix_FT}
\end{figure}
\fi

\end{document}